\newif\ifSC
\SCtrue
\ifSC
\documentclass[onecolumn,draftclsnofoot,12pt]{IEEEtran}
\else
\documentclass[10pt,twocolumn]{IEEEtran}
\fi
\usepackage[normalem]{ulem}
\newcommand{\vehicles}{\mathrm{v}}
\newcommand{\xsec}{\mathrm{x}}
\newcommand{\drisopt}{r_{\ris,\mathrm{o}}}

\newcommand{\drisoptb}{r'_{\ris,\mathrm{o}}}
\newcommand{\distHorzBU}{r_{\ue\bs}}
\newcommand{\distHorzBR}{r_{\ris\bs}}
\renewcommand{\distHorzBR}{r_{\ris}}
\newcommand{\distHorzNN}{r_{\bs\nbr}}

\newcommand{\distHorzUR}{r_{\ue\ris}}

\newcommand{\distBU}{d_{\ue\bs}}
\newcommand{\distBR}{d_{\ris\bs}}
\newcommand{\distUR}{d_{\ue\ris}}
\newcommand{\gBU}{g_{\ue\bs}}

\newcommand{\gUR}{g_{\ue\ris}}

\newcommand{\ellBU}{\ell_{\ue\bs}}
\newcommand{\ellBR}{\ell_{\ris\bs}}
\newcommand{\ellUR}{\ell_{\ue\ris}}
\renewcommand{\ellBU}{\ell_{\bs}}
\renewcommand{\ellBR}{\ell_{\mathrm{d}}}
\renewcommand{\ellUR}{\ell_{\ris}}

\newcommand{\iue}{\mathsf{i}}
\newcommand{\distHorzBIU}{r_{\iue\bs}}

\newcommand{\distHorzIUR}{r_{\iue\ris}}

\newcommand{\E}{\mathsf{E}}
\newcommand{\ptx}{p_\mathrm{t}}
\newcommand{\snrD}{S}
\newcommand{\sinrD}{S_\mathrm{I}}
\newcommand{\snrV}{S'}
\newcommand{\sinrV}{S'_\mathrm{I}}
\newcommand{\intfD}{I}
\newcommand{\intfV}{I'}

\newcommand{\noisen}{\sigma^2}
\newcommand{\noiseP}{N_0}

\newcommand{\pcu}{\mathcal{P}_\UE(\gamma)}
\renewcommand{\pcu}{\mathsf{p}_\UE(\gamma)}
\newcommand{\pcuI}{\mathsf{p}^{'}_{\UE}(\gamma)}

\newcommand{\PhiR}{\Phi_1}
\newcommand{\dR}{d_1}
\newcommand{\PhiL}{\Phi_2}
\newcommand{\dL}{d_2}

\newcommand{\subIIIsection}[1]{\textit{#1}}

\usepackage{amsmath}
\usepackage{graphicx,wrapfig}
\usepackage{amsmath, amssymb, amsthm}
\usepackage{verbatim}
\usepackage{url}
\usepackage[font=footnotesize]{caption}
\usepackage{subcaption}
\usepackage{algcompatible,lipsum}
\usepackage{bm}
\usepackage{bbm}
\usepackage{amssymb}
\usepackage{amsmath,amsthm}
\usepackage{color}
\usepackage{graphicx}
\usepackage{verbatim}
\usepackage{epstopdf}

\newtheorem{theorem}{Theorem}[]
\newtheorem{corollary}{Corollary}[]
\newtheorem{proposition}{Proposition}[]
\newtheorem{lemma}[]{Lemma}
\newtheorem{remark}{Remark}[]

\usepackage[nolist,nohyperlinks]{acronym}
\usepackage{cite}
\usepackage{bbm}
\usepackage{color}

\usepackage{amsmath}
\usepackage{enumitem}

\newcommand{\figspaceadjust}{\vspace{-.21in}}

\newcommand{\jp}[1]{q_{\mathrm{B,S}}(#1)}

\newcommand{\bs}{\mathrm{b}}
\newcommand{\ue}{\mathrm{u}}
\newcommand{\ris}{\mathrm{s}}
\newcommand{\nbr}{\mathrm{n}}
\newcommand{\snrthreshold}{\gamma}
\newcommand{\coefone}{\beta_1}

\renewcommand{\coefone}{\beta}

\newcommand{\sinr}{\SINR}

\newcommand{\A}{ \set{A} }

\newcommand{\dd}{\mathrm{d}}

\newcommand{\fracS}[2]{#1/#2}
\newcommand\expect[1]{\mathbb{E}\left[#1\right]}
\newcommand\prob[1]{\mathbb{P}\left[#1\right]}

\newcommand\indside[1]{\mathbbm{1}\left({#1}\right)}

\newcommand{\SINR}{\text{SINR}}

\newcommand{\expects}[2]{\mathbb{E}_{#1}\left[#2\right] }

\newcommand{\laplace}[1]{\mathcal{L}_{#1} }
\newcommand{\laplaces}[2]{\mathcal{L}_{#1} \left(#2\right)}
\newcommand{\ie}{{\em i.e.}~}

\newcommand{\rx}{\mathrm{R}}

\newcommand{\beam}{\mathrm{b}}

\newcommand{\expU}[1]{e^{#1}}

\newcommand{\expS}[1]{\exp{\left(#1\right)}}

\def\home{\hbox{\kern3pt \vbox to13pt{}%
   \pdfliteral{q 0 0 m 0 5 l 5 10 l 10 5 l 10 0 l 7 0 l 7 5 l 3 5 l 3 0 l f
               1 j 1 J -2 5 m 5 12 l 12 5 l S Q }%
   \kern 13pt}}

\newcommand{\figwidth}{0.470\textwidth}

\renewcommand{\SINR}{\mathtt{SINR}}

\newcommand{\set}[1]{\mathsf{#1}}

 \newcommand{\UE}{\mathrm{U}}
 \renewcommand{\UE}{\ue}

\ifSC
\renewcommand{\figwidth}{0.5\linewidth}
\else
\renewcommand{\figwidth}{0.8\linewidth}
\fi
\newcommand{\figwidthSbS}{0.45\linewidth}

\renewcommand{\figspaceadjust}{}

\renewcommand{\jp}{\mathcal{B}_\ue}
\renewcommand{\jp}{\mathsf{q}_\ue}
\newcommand{\jpD}[1]{\mathsf{q}_\bs(#1) }
\newcommand{\jpV}[1]{\mathsf{q}_{\bs\ris}(#1)}

 \graphicspath{{}{figs/}}
 
\title{On the Deployment of Reconfigurable Intelligent Surfaces in the Presence of  Blockages}

\author{ Vikrant Malik, Gourab Ghatak, Abhishek K. Gupta, Sanket S. Kalamkar 
\thanks{ V. Malik is with the Department of EE, California Institute of Technology, Pasadena, USA. G. Ghatak is with 
the Department of EE, IIT Delhi, India. A. Gupta is with the Department of EE, IIT Kanpur, India. S. S. Kalamkar is with Qualcomm Inc., San Diego, USA.
(E-mail: vmalik@caltech.edu, gghatak@ee.iitd.ac.in, gkrabhi@iitk.ac.in, sanket.kalamkar.work@gmail.com). Part of the paper was presented in IEEE PIMRC 2021 \cite{ghatak2021deploy}.
}}

\begin{document}
\maketitle

\begin{abstract}
    Wireless communications aided by  reconfigurable intelligent surfaces (RISs) is a promising way to improve the coverage for cellular users. The controlled reflection of signals from RISs is especially useful in mm-wave/THz networks when the direct link between a cellular user and its serving base station (BS) is weak or unavailable due to blockages. However, the joint blockage of the user-RIS and the user-BS links may significantly degrade the performance of RIS-aided transmissions. This paper aims to study the impact of joint blockages on downlink performance. When RIS locations are coupled with BS locations, using tools from stochastic geometry, we obtain an optimal placement of RISs to either minimize the joint blockage probability of the user-RIS and the user-BS links or maximize the downlink coverage probability. The results show that installing RISs near the cell edge of BSs usually provides optimal coverage. Moreover, deploying RISs on street intersections improves the coverage probability. For users associated with BSs that are deployed sufficiently close to intersections, the intersection-mounted RISs offer a better coverage performance as compared to BS-coupled RISs.
\end{abstract}

\section{Introduction}
Seamless connectivity, ubiquitous coverage, and high-speed data rates drive the search for new physical layer technologies \cite{saad2019vision,thzchap2021}. A promising step towards achieving these goals is the use of reconfigurable intelligent surfaces (RISs),
realized either by composite meta-materials or by close packing of passive antenna elements~\cite{basar2019wireless}. An RIS can be programmed to function as 
 an anomalous reflector \cite{basar2019wireless,taha2019enabling,bjornson2019intelligent}, where the reflected wave from the RIS can be steered to either enhance signal coverage or localize and track cellular users. In the presence of blockages, RISs can be used to provide an alternate path in case the direct link between a transmitter and a receiver is blocked. However, it is important to optimize the RIS deployment to achieve its full potential. 

\subsection{Related Work}
Since the use of RISs for wireless communications is relatively new, 
the research focus has been restricted to the understanding of its key features and properties in simple networks consisting of a few 
\acp{BS}, users, and RISs~\cite{renzo_tut, nemati2020ris}.  In a large-scale network, the works~\cite{zhang2021multi,zhang2020reconfigurable} have developed  stochastic geometry-based models to evaluate RIS-aided multi-cell networks where multiple RISs are deployed around each user. Note that \ac{mm-wave} communications are significantly impacted by blockages that are an integral part of a dense urban environment~\cite{GupAndHea2017}. Thus it is important to include their impact in the performance analysis \cite{AndrewsmmWaveTut2016, Bai2014}.     
There has been a limited amount of work that includes the blockage effect in the analysis of RISs for large systems. For example, the authors in \cite{zhu2020stochastic} considered the impact of blockages on the user performance with an independent deployment of RISs 
under severe blocking that 
renders all BS-UE links NLOS. 
The authors in \cite{liu2020ris} studied the optimal configurations of RISs for a deterministic channel. 

The deployment of RISs with respect to the BS and the user dictates the end-to-end (i.e., BS to the user via the RIS) path loss, which in turn, decides the communication performance~\cite{zhang_IRS}. In particular, the end-to-end path loss varies as the product of the path losses of the individual (i.e., BS to the RIS and the RIS to the user) links~\cite{ozdogan20pathloss}. Accordingly, it is minimized when an RIS is installed in close proximity to either the BS or the user. However, such a deployment strategy increases the probability that the BS and the RIS are simultaneously blocked from the perspective of the user. Consequently, the RIS cannot be of much use in case the BS is blocked leading to the increased chances of outage. This is due to the fact that blocking of wireless links (the obstruction in the visibility state of the  links) in a network are not independent and there exists a correlation among them. However, past works have not studied the impact of link blockages on the performance of systems with RISs, including the scenarios when there is a correlation between the blockage of the links from a BS and an RIS to the user. This inhibits the derivation of realistic system design insights for \ac{mm-wave} cellular networks. A recent work has characterized the joint blockage probability of two co-existing links in a 2D \ac{mm-wave} cellular network using random shape theory~\cite{GupAndHea2017} and studied its impact on the coverage probability~\cite{GupGup2020GC}. The work in \cite{9369268} has analyzed the joint blockages in RIS-user and BS-user links from the perspective of user localization. However, it does not consider the impact of interference on the  network performance.

Due to the correlation between the blocking of user-BS and user-RIS links, it is crucial to place an RIS in such a way that simultaneous blocking of the two links can be avoided. To achieve this goal, it is important to include blockage correlation in the analysis, which is the main focus of this paper.

\subsection{Contributions}
In this paper, we characterize the performance of a downlink mm-wave cellular network consisting of road-side BSs, aided by RISs.  We focus on deriving the optimal deployment of RISs to improve the network performance in terms of connection reliability and coverage probability. 
This paper makes the following contributions:
\begin{enumerate}
\item We develop an analytical framework for a downlink mm-wave network consisting of RISs-aided-BSs in an urban scenario with multiple streets oriented randomly {in a 2D space} with  
 BSs   deployed along streets. We consider the presence of obstacles on the streets that can block the connections. For each BS, two RISs are installed on each side of the BS 
 to augment downlink transmissions in such cases. 


\item  We focus on addressing a key issue in RIS-aided mm-wave networks: the joint blockages of the user-RIS and the user-BS links which can result in the connection failure. Utilizing the developed framework, we study the impact of blockages on the connection failure probability.

\item Using tools from stochastic geometry, we obtain an analytical expression of the SNR and SINR coverage probabilities while considering the blocking correlation among serving and interfering links. 

\item To minimize the impact of the joint blocking of wireless links on the connection failure probability and SNR/SINR coverage probability, we optimize the deployment strategy for RISs. 

\item 

To further enhance the link reliability, we  extend our analysis to include an additional set of RISs deployed at all intersections of streets. 
\item We provide numerical results highlighting the impact of various system parameters on the performance that offers several insights helpful in the efficient deployment of RISs in the network.
\end{enumerate}

{\textbf{Notation:}
 $\Bar{\E}$ denotes the complementary event of an event $\E$. Unless stated otherwise, $h$ and $r$ are used to denote vertical and horizontal distances, while the Euclidean distances are denoted by $d$. The parameters corresponding to the typical user, its serving BS, and the associated RIS are generally differentiated by subscripts $\ue$, $\bs$, and $\ris$, respectively. Similarly, the intersection user and RIS are denoted by subscripts $\iue$ and $\xsec$, respectively. The notation
 $F_{n}(x)=e^{-x }\sum\nolimits_{q=0}^{n-1}x^{q}/q!$ denotes the complimentary CDF (CCDF) of the Gamma \ac{RV} with parameter $n$.
 }

\section{System Model}
\label{sec:SM}
We consider a downlink mm-wave cellular network consisting of BSs aided by RISs deployed along side streets in a 2D space. The BSs are deployed along streets (e.g., on the lamp posts). 
For each BS, two RISs are installed on each side of the BS to augment downlink transmissions. We consider an urban wireless propagation environment that consists of densely located blockages. 
The system model is as follows.

\textbf{Network geometry.} The street system is modeled by a \ac{PLP} $\Phi_\mathrm{r}$ with intensity $\lambda_\mathrm{r}$ defined on a half-cylinder [$0,2\pi$) $\times$ $\mathbb{R}^+$. The BSs, each of height $h_\mathrm{b}$, are deployed alongside 
the streets. The locations of BSs on the $k$th street  $l_k \in \Phi_\mathrm{r}$ are distributed as points of a one-dimensional \ac{PPP} $\Phi_{\mathrm{b}k}$ with intensity $\lambda_\mathrm{b}$. 
Let $\Phi_\mathrm{b} \triangleq \cup_{k} \Phi_{\mathrm{b}k}$ denote the overall BS point process, which by construction, is a \ac{PLCP}. The direct link between any two points located on different streets $l_k$ and $l_j$, $k \neq j$, is assumed to be blocked by buildings and hence, \ac{NLOS}.

The RIS placement is coupled with BS locations. Specifically, for each BS, two RISs with $N$ elements (or meta-atoms) 
 are deployed at a distance $r_\mathrm{s}$ from the BS on each side of it, and at a height of $h_\mathrm{s}$. We assume that $h_\mathrm{s}$ is always greater than $h_\mathrm{b}$. We consider the typical (general) user  randomly located at the typical street. 
The distance $r_\ris$ between BS and its RISs can be constant or dependent on the cell size of the BS. We analyze these two cases separately.

\textbf{Blockages.} On each street $l_i \in \Phi_\mathrm{r}$, the locations of blockages follow a one-dimensional \ac{PPP} $\Phi_{{\rm v}i}$ with density $\lambda_\mathrm{v}$,  independent of $\Phi_{\mathrm{b}k}$ $\forall k$. Let $h_\mathrm{v}$ denote the height of each blockage. The link between a user and a BS (and similarly the link between a user and an RIS) may be blocked and hence, \ac{NLOS} if a blockage falls on the link. 
A link between a user and  a BS  at distance $r$ is blocked if a blockage is located within a distance of $r\frac{h_\mathrm{v}}{h_\mathrm{b}}$ from the user towards the BS (see Fig. \ref{fig:system}). The probability of this blocking event is given by 
\begin{align}
   p_\mathrm{NLOS}(d)= 1 - \exp\left(- \lambda_\mathrm{v} \frac{h_\mathrm{v}}{h_\mathrm{b}} r\right) \label{eq:blockprob}
\end{align}
due to the void probability of \ac{PPP} \cite{chiu2013stochastic}.
We assume that the received power via a \ac{NLOS} link is negligible~\cite{ghatak2019small}.\footnote{Note that our model can be extended to incorporate the impact of non-negligible \ac{NLOS} signals.}

\begin{figure}[t!]
    \centering
    \includegraphics[width = \figwidth,trim=40 10 40 10, clip]{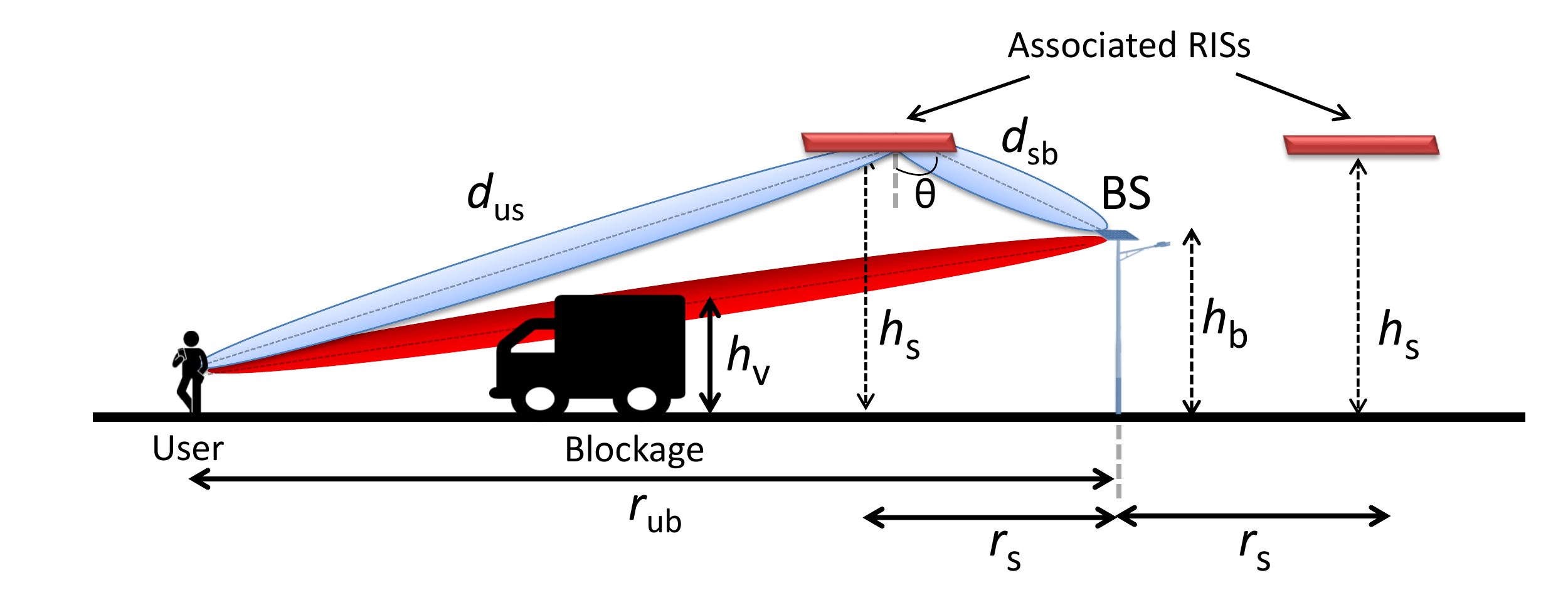}\vspace{-2mm}
    \caption{An illustration of the system model showing one BS and its two RISs serving its user.}
    \label{fig:system}\figspaceadjust
\end{figure}

\textbf{User association.} 
We consider that each user connects to the nearest BS on the same street it is located. The \ac{PDF}  of the horizontal distance $\distHorzBU$ (read as `r user to BS') between the typical  user and the serving/associated BS (which is located at $X_0$) is given as~\cite{chiu2013stochastic}
\begin{align}
    f_{\distHorzBU}(x) = 2\lambda_\mathrm{b} \exp\left(- 2 \lambda_\mathrm{b} x\right), \quad x \ge 0.
    \label{eq:dist_lane}
\end{align}
Without loss of generality, we assume that the associated BS is on the right side of the user.
A BS serves the user via the direct link if it is in the \ac{LOS} state. If the direct link is blocked (\ie is  in the \ac{NLOS} state), the BS uses the reflected link from its RIS located on the same side to serve the user. 
The transmission from the BS to the user experiences an outage if both these links (one from the BS to the user and second from the  RIS to the user) are blocked. 

Let $d_1$ represent the distance of the nearest blockage from the user on the same side as the associated BS. 
On the other hand, the distance of the nearest blockage that is on the opposite side of the associated BS, from the user, is represented by $d_2$. The \ac{PDF}s  of $d_1$ and $d_2$ are the same, and are given by
\begin{align}
    f_{d_i}(x) = \lambda_\mathrm{v} \exp\left(- \lambda_\mathrm{v} x\right), \quad x \ge 0 \quad i \in \{1, 2\}.
    \label{eq:dist_obs}
\end{align}

\textbf{Antenna model.}
The antenna gain patterns of BS and UE are given by $u(\phi)$ and $v(\phi)$, respectively, where $\phi$ is the angle between the antenna orientation and the desired direction. The transmit antenna gain of the $i$th BS is denoted by an \ac{RV} 
$U_i$. Similarly, the receive antenna gain for the same BS is denoted by an \ac{RV} $V_i$. Since each street can be represented by a 1D line, 
 the antenna can point at either side of itself, \ie, 
\begin{align}
u(\phi) \!= \!
    \begin{cases}
    u_m &\!\!\!\! \text{towards the antenna direction}\\
    u_s &\!\!\!\! \text{otherwise}    \end{cases},
    &&
    v(\phi)\!=\!
    \begin{cases}
     v_m &\!\!\! \text{towards the antenna direction.}\\
     v_s &\!\!\! \text{otherwise}.
    \end{cases} \label{eq:excl}
\end{align}

\textbf{Path loss model.}
The instantaneous received power $P_{{\rm r},d}$ at the receiver from an \ac{LOS} transmitter located at distance $d$, is $P_{{\rm r},x} = \ptx\, g_{\rm tr}\,u(\phi) v(\phi) \ell(d)$, where $\ptx$ is the transmit power of a BS and $\ell(d)$ is the path-loss.  $u(\phi)$ and $v(\phi)$ are the transmit and receive beam-forming gains. 
We consider Nakagami fading, \ie, the fading power $g_{\rm tr}$ is Gamma distributed with parameter $n_0$. 
 One example is the standard power-law path-loss model given as
 $\ell(d)= K(d)^{-\alpha}$, where 
 $K = (\frac{\lambda}{4 \pi})^2$ is a constant and $\alpha$ is  the path-loss exponent.
Let $\coefone \triangleq K u_m v_m$. 
Let us also define the following specific path-loss functions
\begin{align*}
\ellBU(r)=\ell(\sqrt{r^2+h^2_\bs})&&
\ellUR(r)=\ell(\sqrt{r^2+h^2_\ris})&&
\ellBR(r)=\ell(\sqrt{r^2+(h_\ris-h_\bs)^2}).
\end{align*}

\textbf{SNR/SINR model for direct link.} 
If the link between the user and the serving BS is \ac{LOS}, the  serving received power $P_{\rm r}$ at this user is $P_{\rm r} = \ptx\, \gBU\,u_m v_m \ell (\distBU)$, where $\distBU=\sqrt {\distHorzBU^2 + h_\mathrm{b}^2}$ is the  distance of the BS to the user and $g_{\rm{bu}}$ denotes the fading power gain on the link between the user and the BS. 
Hence, the user downlink \ac{SNR}  is
\begin{align}
    \snrD = \frac{P_{\rm r}}{\noisen}= \frac{\gBU u_m v_m\ellBU(\distHorzBU) }{\noisen}.
\end{align}
Here, $\noisen=\noiseP/\ptx$ is the noise power. 
Similarly, the SINR, in this case, can be given by
\begin{align}
    \sinrD 
    = 
    \frac{\ellBU(\distHorzBU) g_{\rm{bu}}u_m v_m}{\noisen + \intfD},
    \label{eq:SI}
\end{align}
where $\intfD$ is the interference. 

\textbf{SNR/SINR model for via-RIS link.} 
As mentioned before, if the link from the  user to the BS is blocked and the RIS is in \ac{LOS}, the user is served by the BS via RIS link. The signal power at the user due to this link is given as \cite{kokkoniemi2021stochastic}
\begin{align}
    P_r'  = 
    \gUR  {\ptx u_m  \ellBR(\distHorzBR) \pi (N - 1)^2 } 
    { \ellUR(\distHorzUR) v_m},
    \label{eq:viaRIS1}
\end{align}
where $\distHorzUR $ is the horizontal distance of the user from the RIS, $\distHorzBR $ is the horizontal distance of the RIS from the BS, $\gUR$ is the fading power gain (which is gamma distributed with parameter $n_0$) on the link between the RIS and the user, and $\theta$ is the incidence angle at RIS. 
If we define \begin{align}
\ell_\mathrm{r}(\distHorzBU, r_s)=  \pi  (N - 1)^2   \ellBR(\distHorzBR) \ellUR(\distHorzUR),
\label{eq:lr}
\end{align}
 \eqref{eq:viaRIS1} can be written as
\begin{align}
    P_r'  
    = { \ell_\mathrm{r}(\distHorzBU, r_s) \gUR u_m v_m} .
\end{align}
%
In \eqref{eq:lr}, 
$\distHorzUR=|\distHorzBU-\distHorzBR|$. 
Note that for standard path-loss with $\alpha=2$, 
    $P_r'  = 
     \ptx \coefone \pi K (N - 1)^2  \left({\distUR \distBR }\right)^{-2} \gUR $ with $\distBR  = \sqrt{r_\mathrm{s}^2 + (h_\mathrm{s} - h_\mathrm{b})^2}$ and $\distUR  = \sqrt{(\distHorzBU- r_\mathrm{s})^2 + h_\mathrm{s}^2}$.
Now, the \ac{SNR} and \ac{SINR} at the user are respectively given as 
\begin{align}
    \snrV
    &= \frac{ \ell_\mathrm{r}(\distHorzBU, r_s) g_{\rm{su}}u_m v_m}{\noisen},&
    \sinrV 
    &= \frac{ \ell_{\mathrm{r}}(\distHorzBU, r_s) g_{\rm{su}}u_m v_m}{\noisen + \intfV}, 
    \label{eq:SIp}
\end{align}
where $\intfV$ is the interference. 

\textbf{Performance metrics}. We consider the following  metrics to evaluate the performance of the considered system. 

\textit{Connection failure probability} $\jp{}$ is the probability that the typical user does not have an \ac{LOS} connection from the serving BS and its RIS.
Let $\E_{\rm b}$ and $\E_\mathrm{s}$ denote the events that the links from the user to the BS and the RIS are blocked, respectively. Then, $\jp{}=
    \mathbb{P}(\E_\mathrm{b} \cap \E_\mathrm{s}) $.
 
 \textit{SNR/SINR coverage probability} is defined as the probability that the SNR/SINR of the user is greater than the threshold $\gamma$, \ie, $\mathsf{p}_\mathrm{c}(\gamma)=\prob{\mathtt{SNR}>\gamma}$ and 
$\mathsf{p}'_\mathrm{c}(\gamma)=\prob{\sinr>\gamma}$.

In the next sections, we analyze the performance of the considered system in terms of the aforementioned performance metrics for various RIS deployments including constant and adaptive $r_\ris$. 


\section{Analysis for Fixed Distance Deployment of RISs}
In this section, we consider the case where the distance $r_\ris$ between BS and its RISs is constant. 

\subsection{Characterization of Serving Links Blockings and Connection Failure Probability}

\label{sec:LB}
We now compute the connection failure probability. From \eqref{eq:blockprob}, 
we have the following result.
\begin{lemma}\label{lemma:BSBPGivenrud}
The probability that the link between the typical user and the serving BS is blocked conditioned on horizontal distance $\distHorzBU$ between them is given by~\cite{chiu2013stochastic}
\begin{align}
    \jpD{\distHorzBU}\stackrel{\Delta}{=}
    \mathbb{P}(\E_\bs\mid \distHorzBU) = 1 - \exp\left(- \lambda_\mathrm{v} \distHorzBU{h_\mathrm{v}}/{h_\mathrm{b}}\right). \label{eq:baseblock}
\end{align}
\end{lemma}%
Recall that the BS switches its beam towards the RIS if its direct link to the user is blocked. Thus, the user is in a complete signal outage when the direct and the RIS links are simultaneously blocked. 
\begin{lemma}\label{lemma:JBPGivenrud}
Conditioned on the horizontal distance $\distHorzBU$ between the typical  user and the BS, the joint probability of the events $\E_\mathrm{s}$ and $\E_\mathrm{b}$ is given by
\begin{align*}
 &\jpV{\distHorzBU}=\mathbb{P}(\E_\mathrm{b}\cap \E_\mathrm{s}\mid \distHorzBU) = 
 \begin{cases}
    \left(1 - 
    \expU{-\lambda_\mathrm{v}\frac{h_\mathrm{v}}{h_\mathrm{b}}\distHorzBU}\right)  \cdot 
    \left(1 - 
    \expU{- 
    	\lambda_\mathrm{v} 
	(r_\mathrm{s} - \distHorzBU)\frac{h_\mathrm{v}}{h_\mathrm{s}}}
    \right)& \text{\rm if }\   0 \leq \distHorzBU \leq r_\mathrm{s}, \\
    1 - \exp\left(- \lambda_\mathrm{v} (\distHorzBU - r_\mathrm{s})\frac{h_\mathrm{v}}{h_\mathrm{s}}\right)& \text{\rm if }\  \distHorzBU > r_\mathrm{s}.
    \end{cases}
\end{align*}
\end{lemma}
\begin{IEEEproof}
For $0 \leq \distHorzBU \leq r_\mathrm{s}$, the typical  user is located between the serving BS and the RIS. Accordingly, the joint blockage probability is the product of the blockage probabilities of individual links from the user to the BS and the RIS. On the contrary, for $\distHorzBU \geq r_\mathrm{s}$, the user-BS and the user-RIS links are jointly blocked when a blockage is located at a distance of $\distHorzBU - r_\mathrm{s}$ from the user. The result follows from the void probability of the PPP \cite{AndGupDhi16}.
\end{IEEEproof}

\begin{remark}
Lemma \ref{lemma:JBPGivenrud} shows that if an RIS is placed very close or very far from the BS, the joint blockage probability is high. For close deployments of the RIS to the BS, the user-BS link is blocked whenever the user-RIS link is blocked. On the contrary, when an RIS is deployed far from the BS, the individual links from the user to the BS, and the RIS have a high probability of blockage. In particular, for a given $\distHorzBU$, the joint blockage probability is minimum when $r_\ris=\distHorzBU$. This means that from a particular user's perspective, the RIS should be located directly above the user. Since there are many users, one should select a location of the RIS that results in the smallest blockage probability on average.
\end{remark}

By unconditioning over $r_\mathrm{bu}$, we obtain the {average} joint blockage probability as
$
 \jp{}=\mathbb{P}(\E_\mathrm{b}, \E_\mathrm{s}) = \int_0^\infty \mathbb{P}(\E_\mathrm{b}, \E_\mathrm{s}\mid \distHorzBU=x) f_{\distHorzBU}(x)  \mathrm{d}x \nonumber.
$
Using $f_{\distHorzBU}(x) $  given in \eqref{eq:dist_lane}, we get
 the following Proposition. 

\newcommand{\rsN}{f_\ris}
\newcommand{\bseverity}{\mathrm{R_b} }
\newcommand{\rseverity}{\mathrm{R_s} }
\newcommand{\bseverityinv}{1/\mathrm{R_1} }
\newcommand{\rseverityinv}{1/\mathrm{R_2} }
\renewcommand{\bseverityinv}{\rho_\bs }
\renewcommand{\rseverityinv}{\rho_\ris}

\begin{proposition}
The connection failure probability 
is given by
\begin{align}
    \jp{}=
    &1 - 2
    \left[
    \frac{
                	\exp\left(-\rsN
                \left(
                \bseverityinv
                + 2\right)\right) - 
                \exp\left(-
                \rsN \rseverityinv
                 \right)
            }
     	  {
	    	2 + 
     		\bseverityinv
      		-\rseverityinv
      	   } 
      + \frac{\exp(-2
      \rsN
      )}{
      \rseverityinv
       + 2}  \nonumber \right.\\
    &\left. + 
    \frac{1 - \exp\left(-\rsN 
    \left(
  	\bseverityinv
    + 2\right)\right) }{
    \bseverityinv
    + 
    2} + 
    \frac{\exp\left(
    	-2 \rsN 
	\right) - \exp\left(-
	\rsN \rseverityinv 
	 \right)}{
	 \rseverityinv
	 - 2} \right], \label{eq:joint_block}
\end{align}
where $\bseverityinv =1/\bseverity= \frac{\lambda_\mathrm{v} h_\mathrm{v}}{\lambda_\mathrm{b} h_\mathrm{b}}$ and $\rseverityinv=1/\rseverity= \frac{\lambda_\mathrm{v} h_\mathrm{v}}{\lambda_\mathrm{b} h_\mathrm{s}}$. Here, $\bseverityinv$ and $\rseverityinv$ denote the relative severity of blockages with respect to BS and RIS deployment.
$\rsN=\lambda_\bs r_\ris$ denotes the RIS distance relative to average BS cell size.
\end{proposition}


\begin{proposition}\label{pro:block_bound}
 $\jp{}$ can be bounded as
 $\displaystyle
    \frac{
    		e^{-2 \rsN
    	}}
    	{1+ 2 \bseverity
	} 
	\le \jp{} \leq 
	\left(1-e^{-2
	\rsN
	}\right)+ 
    \frac{e^{-2
    \rsN 
    }}{1+2 \bseverity
    }. \nonumber
$
\end{proposition}


\begin{remark}
The bounds in Proposition 2 become strict when $r_{\rm s}$ is small, i.e., $r_{\rm s} \ll 1/\lambda_\bs$.
\end{remark}

\begin{lemma}\label{lemma:optrsblock}
The optimal RIS distance \textcolor{black}{for minimizing the joint blockage probability} is given by
\begin{align*}
    \drisopt
   = {(\lambda_\bs \bseverityinv)}^{-1} 
   \ln{{x_\mathrm{o}^{-1}}},
\end{align*}
where $x_\mathrm{o}$ is the solution of the Lambert's equation
\begin{align}
    x^a=a x +c,\ x\in[0,\ 1]\label{eq:xopt}, && \text{ with }
%
a=\frac{(\rseverityinv-2)}{\bseverityinv},
c=4\frac{(1-a)}{(\rseverityinv+2)},
\end{align}
with $a<1$ and $0<c$. 
The approximate solution is given as
\begin{align*}
    r_\ris=\frac1{\lambda_\bs}\sqrt{\frac2{
    (\bseverityinv) (\rseverityinv+2{})}}=\sqrt{\frac2{ \frac{\lambda_\vehicles h_\vehicles}{h_\bs}(\frac{\lambda_\vehicles h_\vehicles}{h_\ris}+2{\lambda_\bs})}}.
\end{align*}

\end{lemma}
\begin{proof}
See Appendix \ref{app:optrsblock}.
\end{proof}
It can be seen that with {increasing} blockage severity (\ie, with an increase in $\lambda_\vehicles$ or $h_\vehicles$), {the} optimal distance of $r_\ris$ decreases.

\begin{corollary}
[Asymptotic case]
As $\lambda_\vehicles\rightarrow \infty$, $\bseverityinv,\rseverityinv\rightarrow\infty$, and accordingly, $a=\rseverityinv/\bseverityinv$ is constant and $c=0$. Therefore, from \eqref{eq:xopt}, {we have} 
\begin{align*}
    x^{a}=ax \implies & (a-1)\ln(x)=\ln(a).& \text{Hence, } 
    & \drisopt = \frac1{\lambda_\bs \bseverityinv}\frac{\ln(a)}{1-a}
    \rightarrow 0.
\end{align*}
\end{corollary}

\subsection{Association Probability}
Let $\A_\bs$ denote the probability that the user is served by direct link, \ie, the serving BS is LOS. Similarly, let $\A_\ris$ denote the probability that the user is served via RIS, \ie, the serving BS is NLOS and the associated RIS is LOS. Hence, we get
\begin{align}
\prob{\A_\bs}=\prob{\bar{\E_\bs}}=1-\expect{\jpD{\distHorzBU}}, \prob{\A_\ris}=\prob{{\E_\bs} \cap \bar{ \E_\ris}}=\prob{{\E_\bs}}-\prob{\E_\bs \cap { \E_\ris}}=
\expect{\jpD{\distHorzBU}}-\jp{}.\nonumber
\end{align}


\subsection{SNR Coverage Probability}
\label{sec:SNR}
We will now compute the coverage probability of the typical user for threshold $\gamma$ 
which is given in the following theorem. 
\begin{theorem}\label{thm1}
The SNR coverage probability of the typical 
 user is (See Appendix \ref{app:thm1}.)
\begin{align}
    \pcu= \mathbb{E}_{\distHorzBU}\left[
    (1- \jpD{\distHorzBU})
    F_{n_0}\left(
    \frac
    {\gamma \noisen 
    }
    { \ellBU(\distHorzBU)u_m v_m
    }
    \right) \right. 
    +
    \left.
    (\jpD{\distHorzBU}-\jpV{\distHorzBU})
    F_{n_0}\left(
    \frac{\gamma \noisen
    }{\ell_\mathrm{r}(\distHorzBU, r_{\rm s})u_m v_m
    }\right)   \right] \nonumber .
\end{align}
\end{theorem}

According to Alzer's lemma~\cite{alzer1997some,AndrewsmmWaveTut2016}, the CCDF of a Gamma \ac{RV} $X$  with parameter $n_0$ can be {upper} bounded as
\begin{align}
    F_{n_0}(\gamma) =\mathbb{P}(X \geq \gamma) \leq
    \sum\limits_{n = 1}^{n_0}(-1)^{n+1} \binom{n_0}{n} 
    e^{-n \gamma \eta_m }
    \label{eq:alzerlemma}
\end{align}
with $\eta_m = n_0(n_0!)^{-\frac{1}{n_0}}$. 
Using \eqref{eq:alzerlemma} in Theorem \ref{thm1}, we get the desired result.
\begin{corollary}
The SNR coverage probability of the typical user can be approximated as
\begin{align}
    \pcu 
    \approx \;\sum\limits_{n = 1}^{n_0}(-1)^{n+1}\binom{n_0}{n}
    \mathbb{E}_{\distHorzBU}&\left[ \exp\left(
    {-\frac{n \gamma'  }{\ellBU(\distHorzBU)} 
    }
    {
    }\right)  
    (1-\jpD{\distHorzBU})
    \nonumber \right. \\
    & \ \ \ +\left.
    \exp\left(-\frac{n \gamma'  
    }
    {
    \ell_\mathrm{r}(\distHorzBU, r_{\rm s})}\right) 
    (\jpD{\distHorzBU}-\jpV{\distHorzBU})
    \right],
\end{align}
where
$\gamma' =\eta_m \snrthreshold \noisen  /(u_mv_m) 
$. 
\end{corollary}

\begin{corollary}
 {For the special case of standard path-loss model with $\alpha = 2$}, the SNR coverage probability can be simplified as
\begin{align}
    \pcu \approx &\sum\limits_{n = 1}^{n_0}(-1)^{n+1} \binom{n_0}{n} \left[\exp\left(\!\frac1{n\gamma'} \lambda_\bs^2 {\left(1\!+\!
    \frac12\bseverityinv
    \right)^2}\!\!\! -\! n\gamma' h_\mathrm{b}^2 \!\right)  
    \nonumber  
    \!
    \textnormal{erfc}\left(
    \frac{\lambda_\bs(1+\bseverityinv/2)}{\sqrt{n\gamma'}}\right)\right.\!
    \frac{\lambda_\mathrm{b} \sqrt{\pi}}{\sqrt{n \gamma'}} \nonumber 
    \nonumber \\
    &+ \left.\mathbb{E}_{\distHorzBU}\left[\exp\left(\!
    -\frac{
    n \gamma' 
    }{
    \ell_\mathrm{r}(\distHorzBU, r_{\rm s})}
    \right)
     (\jpD{\distHorzBU}-\jpV{\distHorzBU})
    \vphantom{\frac1\xi}
    \right]\right].\nonumber
\end{align}
\end{corollary}
In the absence of blockages, the user to BS link is never blocked, and we get the following result.
\begin{corollary}
For the special case of $\alpha = 2$ and with the absence of blockages, the coverage probability {can be given as
\begin{equation*}
    \pcu=\lambda_\mathrm{b} \sqrt{\pi} \sum\limits_{n = 1}^{n_0} \binom{n_0}{n} 
    \frac{(-1)^{n+1}}{\sqrt{n \gamma'}}
    {\exp\left(\frac{\lambda_\mathrm{b}^2}{n \gamma'} - n \gamma' h_\mathrm{b}^2 \right)\mathrm{erfc}\left(\frac{\lambda_\mathrm{b}}{\sqrt{n \gamma'}}\right)}.
\end{equation*}}
\end{corollary}

\subsection{Interference Characterization}
We now characterize the interference at the user via its Laplace transform. 
Now, if a BS at a given distance is blocked, then all other BSs at a larger distance on the same side will also be blocked, which results in a correlation of blocking states of the BSs on the same side~\cite{GupMalGupAnd2022}. To handle this correlation, we analyze interference conditioned on the distances $\dL$ and $\dR$ of the nearest blockages. 
Since interference is $\intfD$ or $\intfV$ depending on the serving link, we handle these {cases} separately. 

\subsubsection{Direct serving link}
When $\dR> \distHorzBU \fracS{ h_{\vehicles}}{h_{\bs}}$, the serving BS of the user is \ac{LOS}, \ie, the event 
\begin{align}
\A_\bs= \{\dR> \distHorzBU \fracS{ h_{\vehicles}}{h_{\bs}}\},\label{eq:Ab}
\end{align}
{occurs.} In this case, the user is served via a direct link by the serving BS. The user receives interference from all \ac{LOS} BSs located at both sides. 
Further,  for a BS located at distance $r$ on the right side of the user to be \ac{LOS}, we have $\dR > \fracS{r h_{v}}{h_{b}} $. This means that BSs located at distances {only up to $\frac{\dR h_\bs}{h_\vehicles}$} will interfere. Similarly, for BSs at left side, $\dL > \fracS{r h_{v}}{h_{b}} $.
%
%
The  interference is given as
\begin{align}
I(\Phi)&=I(\PhiR)+I(\PhiL)\\
I(\PhiR)&=\sum\nolimits_{X_{i} \in \PhiR \backslash\left\{X_{0}\right\}} g_{i} \ell(\|X_i\|)U_i v_m \cdot \mathbbm{1} \left(\|X_i\| < \fracS{\dR h_\bs}{h_{\vehicles}} \right)\label{eq:IRFD}\\
I(\PhiL)&=\sum\nolimits_{X_{i} \in \PhiL } g_{i} \ell(\|X_i\|)U_i v_s \cdot \mathbbm{1} \left(\|X_i\| < \fracS{\dL h_\bs}{h_{\vehicles}} \right).
\end{align}

\begin{lemma}\label{lem:intfD}
The Laplace transform of the interference from BSs when the user is served by a direct link is given by (See Appendix \ref{app:lap})
\begin{align}
\mathcal{L}_{I}(s|\distHorzBU,\dR,\dL )= &
\expS{ -{\lambda_\bs}
 \int_{\distHorzBU}^{\max \left(\distHorzBU, \fracS{\dR h_\bs}{h_{\vehicles}}\right)}
 \left( \textstyle
 1- \frac{1}{2\left(1+ s \ellBU(x)u_m v_m\right)^{n_0}} -\frac{1}{2\left(1+ s \ellBU(x)u_s v_m\right)^{n_0}} \right)\dd x \nonumber  }\\
 &\times \expS{ -{\lambda_\bs}
\int_{\distHorzBU}^{\max \left(\distHorzBU, \fracS{\dL h_\bs}{h_{\vehicles}}\right)} \left( \textstyle
 1- \frac{1}{2\left(1+ s \ellBU(x)u_m v_s\right)^{n_0}} -\frac{1}{2\left(1+ s \ellBU(x)u_s v_s\right)^{n_0}}\right) \dd x   }.\nonumber
\end{align}
\end{lemma}


Special case: For Rayleigh fading $(n_0=1)$,
\begin{align}
&\mathcal{L}_{I}(s| \distHorzBU,\dR,\dL) \\
	&= \exp \left( -\frac{\lambda_\bs}{2}
	 \left(\int_{r}^{\max \left(\distHorzBU, \frac{d_{1} h_{b}}{h_{v}}\right)}
	  \frac{s \ellBU(x)u_m v_m}{1+ s \ellBU(x)u_m v_m} \dd x +
	   \int_{\distHorzBU}^{\max \left(\distHorzBU, \frac{d_{1} h_{b}}{h_{v}}\right)} 
		\frac{s \ellBU(x)u_s v_m}{1+ s \ellBU(x)u_s v_m} \dd x \nonumber \right.\right. \\
	&\qquad
	\left. \left. +  \int_{\distHorzBU}^{\max \left(r, \frac{d_{2} h_{b}}{h_{v}}\right)} 
		\frac{s \ellBU(x)u_m v_s}{1+ s \ellBU(x)u_m v_s} \dd x + 
	\int_{\distHorzBU}^{\max \left(r, \frac{d_{2} h_{b}}{h_{v}}\right)} \frac{s \ellBU(x)u_s v_s}
	{1+ s \ellBU(x)u_s v_s} \dd x \right) \right).\nonumber
\end{align}

\subsubsection{Via-RIS link}
{Note that} when $d_1 < \fracS{\distHorzBU h_\vehicles}{h_\bs}$, the serving BS is NLOS.
If the RIS is in LOS,  the user is  served via a RIS. 
The blocking state of the RIS link depends on $\distHorzBU$ as explained below.

\textbf{Case-I:} when $\distHorzBU< r_\ris$, the associated RIS is on the left side of the user. The blocking state of the UE-RIS link  depends on the blockage distance $\dL$. In particular, this link is LOS when {$\dL > \frac{(r_\ris-\distHorzBU)h_{\vehicles}}{h_{\ris}}$.}


\textbf{Case-II:} when $\distHorzBU> r_\ris$,  the associated RIS is on the right side of the user. Hence, the blocking of the UE-RIS link depends  on $\dR$. This link is LOS when {$\dR > \frac{(\distHorzBU -r_\ris)h_{\vehicles}}{h_{\ris}}$.}


Combining the two, we get that the event that the user is served via RIS is equivalent to
\begin{align}
\A_\ris=\E_\bs\cap\bar{\E_\ris}=\begin{cases}
d_1 < \fracS{\distHorzBU h_\vehicles}{h_\bs}, \ (r_\ris-\distHorzBU ) \frac{h_{\vehicles}}{h_{\ris}}< \dL & \text{ if } \distHorzBU< r_\ris\\
d_1 < \fracS{\distHorzBU h_\vehicles}{h_\bs}, \ (\distHorzBU -r_\ris)\frac{h_{\vehicles}}{h_{\ris}}< \dR & \text{ if } \distHorzBU> r_\ris.
\end{cases}\label{eq:As}
\end{align}

In this case, since the BS is blocked, all other BSs from the same side are also NLOS and hence, the user receives interference from BSs located on only one side (\ie, the left side). We denote it by $\intfV_1$. 
Additionally, a certain amount of interference $\intfV_2$ can come from the same side BSs via RIS. We study both in the following. 

\subIIIsection{BS interference $\intfV_1$:} The interference from BSs is given as
\begin{align*}
    \intfV_1 &= \sum_{X_{i} \in \PhiL \backslash\left\{B_{0}\right\}} 
    g_{i}U_i V_i   \ellBU(\|X_i\|)
	 \cdot \indside{\|X_i\| < 
			\frac{\dL h_\bs}
			{h_\vehicles}}.
\end{align*}
Here, the receive antenna gain $V_i$ will again depend on $\distHorzBU$ as explained below

\textbf{Case-I:} when $\distHorzBU< r_\ris$, the associated RIS is on the left side of the user, and hence, the user receiver antenna will point towards the left. This means that the interfering BSs are in the  direction of the main lobe of the receiving antenna.  Hence, $V_i=v_m$ for all BSs.

\textbf{Case-II:} when $\distHorzBU> r_\ris$,  the associated RIS is on the right side of the user. Hence, the interfering BSs are in the back lobe direction. Hence, $V_i=v_s$ for all BSs.

Now, the Laplace transform can be written as,
\begin{align}
&{\mathcal{L}_{\intfV_1}(s | \distHorzBU, \dL)} = \mathbb{E}_{\Phi, g_{i}, U}\left[\exp \left(-s I_1 \right)\right] =\\
& \expS{\! -\frac{\lambda_{b}}{2} \left(\int_{\distHorzBU}^{\infty} \!
\left[2 - \frac{1}{\left(1+ s \ellBU(x)u_m v(\distHorzBU)\right)^{n_0}}- \frac{1}{\left(1+ s \ellBU(x)u_s v(\distHorzBU)
\right)^{n_0}} \right]
\indside{x < \frac{\dL h_{\bs}}{h_{\vehicles}}\!}
\dd x \right)}\label{eq:intfV1}
\end{align}
with $v(r)=v_m \cdot \mathbbm{1} \left(r < r_s\right) + v_s \cdot \mathbbm{1} \left(r > r_s\right) $.

\subIIIsection{RIS Interference $\intfV_2$:}
Since RISs can receive signals from BSs other than the tagged BSs, they can unintentionally reflect these waves to the user which can cause interference degrading the performance. Moreover, RISs can reflect very narrow beams to the selected user and hence, it is very unlikely that there will be interference at a user from other RISs. 
Let us  assume that the RIS reflective beamwidth $\theta_\ris$ is very small. Let $z_\ris$ be the {distance from RIS upto which a BS can interfere}, {\em i.e.},
   $ z_\ris=\sqrt{\left(h_\ris - h_\bs\right)^{2}+r_{\ris}^{2}} \cdot \theta_{\ris}$.
We can see from Fig. \ref{fig:SINR_BS} that a BS will cause interference at the user (via RIS) only if it lies within a particular region 
it lies within distance $z$ from the tagged BS (on the right). 
The interference power from each BS via RIS is given as
\begin{align}
    P_r'  = 
    h_i  { U_i  \ellBR(x-(\distHorzBU-r_\ris)) \pi (N - 1)^2 } 
    { \ellUR(\distHorzUR) v_m}=h_i \ell_\mathrm{rI}(x,\distHorzUR,r_\ris)U_i v_m,
\end{align}
where $\ell_\mathrm{rI}(x,\distHorzUR,r_\ris)=\ellBR(x-\distHorzUR) \pi (N - 1)^2  
    { \ellUR(\distHorzUR)}$.
Hence, the sum interference is given as
\begin{align*}
    \intfV_2 
     &=\sum_{\substack{X_{i} \in \PhiR \backslash\left\{X_{0}\right\} }} h_{i}U_i v_m \ell_\mathrm{rI}(X_i,\distHorzUR,r_\ris)\indside{X_i<\distHorzBU+z_\ris}.
\end{align*}
\begin{figure}[ht!]
    \centering
    \includegraphics[width = 0.8\textwidth,trim=0 20 0 35,clip]{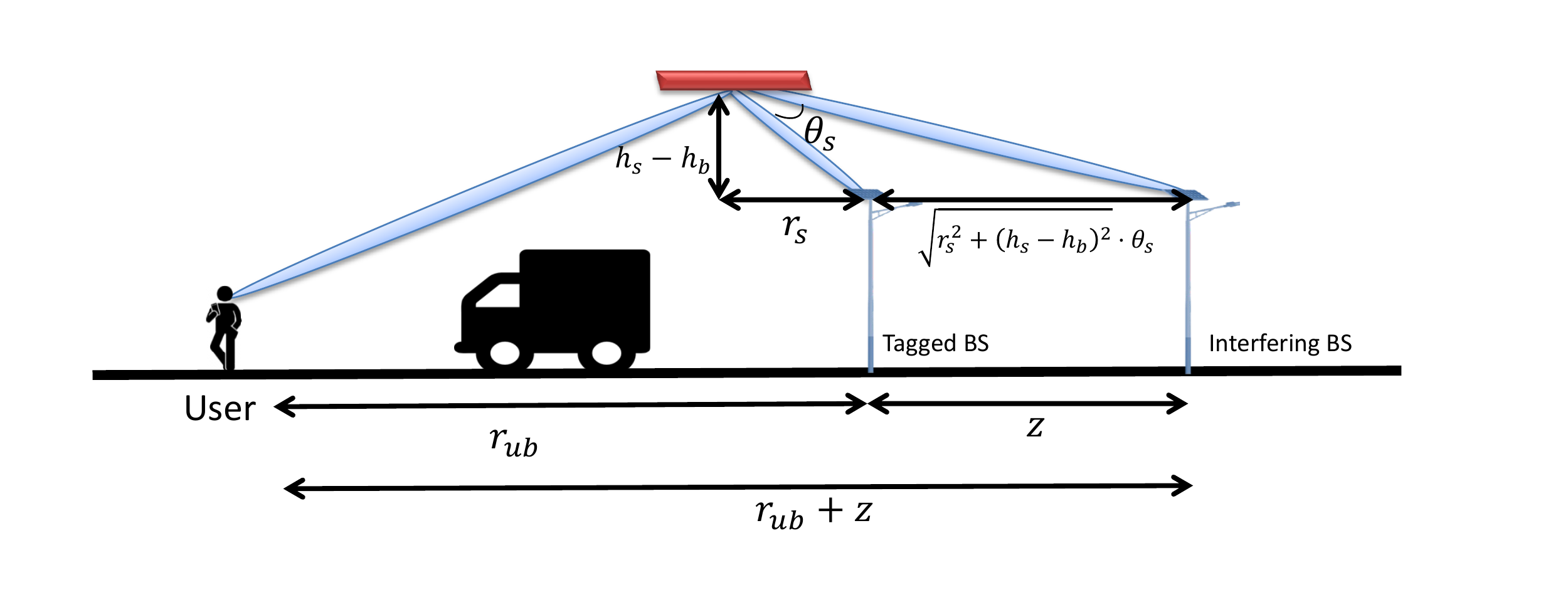}
    \caption{Illustration showing the  BSs interfering via RIS. BSs are located in a small region on the right of the serving BS.}
    \label{fig:SINR_BS}
\figspaceadjust
\end{figure}

The Laplace transform of the interference is given as
\begin{align}
{\laplace{\intfV_2}(s | \distHorzBU)} = \expU{ -\frac{\lambda_{b}}{2} \left(\int_{\distHorzBU}^{\distHorzBU+z_\ris} \left(2 - \frac{1}{\left(1 + s \ell_{\mathrm{rI}}(x)u_m v_m \right)^{n_0}}  - \frac{1}{\left(1 + s \ell_{\mathrm{rI}}(x)u_s v_m \right)^{n_0}}\right) \dd x \right)}.\label{eq:intfV2}
\end{align}
Since the total interference $\intfV$ is given as $\intfV=\intfV_1+\intfV_2$, we get the following result.

\begin{lemma}\label{lem:intfV}
The  Laplace transform of the  interference  when the user is served via a RIS is 
\begin{align*}
{\laplace{\intfV}(s \mid \distHorzBU, d_2) = \laplace{\intfV_1}(s \mid \distHorzBU, d_2)\laplace{\intfV_2}(s \mid \distHorzBU),}
\end{align*}
where the two terms are given in \eqref{eq:intfV1} and \eqref{eq:intfV2}.
\end{lemma}

\subsection{SINR Coverage probability}

We will now present the SINR coverage of the typical user in the following theorem that includes the effect of various interference.

\begin{theorem}\label{thm:PcIFixed}
The SINR coverage probability of the typical user is given as
\begin{align*}
&\pcuI
=\sum_{n=1}^{n_0}
			(-1)^{n+1} \binom{n_0}{n} \bigg[\\\
			&{\int_0^\infty\int_0^\infty\int_{\distHorzBU h_\vehicles/h_\bs}^\infty}
			e^{- \frac{ n \gamma \noisen \eta_m}{\ellBU(\distHorzBU)u_m v_m}  }
		\laplaces{I|\distHorzBU,\dR,\dL}{\textstyle \frac{n\gamma\eta_m  }{\ellBU(\distHorzBU)u_m v_m} }
 f_{\dR}(\dR)  f_{\dL}(\dL)f_{\distHorzBU}(\distHorzBU) \dd\dR\dd\dL\dd \distHorzBU\\
&+\int_{0}^{r_\ris}{ \jpD{\distHorzBU} \cdot \int_{(r_\ris-\distHorzBU) h_\vehicles/h_\ris}^\infty 
e^{- \frac{ n \gamma \noisen \eta_m}{\ell_\mathrm{r}(\distHorzBU,r_\ris)u_m v_m}  }
		\laplaces{\intfV|\distHorzBU,\dL}{\textstyle \frac{n\gamma\eta_m  }{\ell_\mathrm{r}(\distHorzBU,r_\ris)u_m v_m} }
  f_{\dL}(\dL)f_{\distHorzBU}(\distHorzBU)\dd\dL\dd \distHorzBU}\\
&+\int_{r_\ris}^\infty{ \jpD{\distHorzBU} \cdot \int_{(\distHorzBU-r_\ris) h_\vehicles/h_\ris}^{\distHorzBU h_\vehicles/h_\bs}
e^{- \frac{ n \gamma \noisen \eta_m}{\ell_\mathrm{r}(\distHorzBU,r_\ris)u_m v_m}  }
		\laplaces{\intfV|\distHorzBU,\dL}{\textstyle \frac{n\gamma\eta_m  }{\ell_\mathrm{r}(\distHorzBU,r_\ris)u_m v_m} }
 f_{\dL}(\dL)f_{\distHorzBU}(\distHorzBU) \dd\dL\dd \distHorzBU
\bigg],}\end{align*}
where $\laplace{\intfD}(s)$ and $\laplaces{\intfV}{s}$ are given in Lemma \ref{lem:intfD} and \ref{lem:intfV}.
\end{theorem}
 \begin{IEEEproof}
Similar to previous sections, to account for the correlation of blocking states of the BSs on the same side \cite{GupMalGupAnd2022}, we first analyze coverage probability conditioned on the distances $\dL$ and $\dR$ of the nearest blockages and then average it with respect to their distributions. See Appendix \ref{app:PcIFixed} for detailed proof.
\end{IEEEproof}

The coverage is a function of $r_\ris$ and hence, can be maximized by choosing an optimal value of $r_\ris$, as we do in the numerical section. 

\section{Analysis for Cell-Size Dependent Deployment of RISs}
The fixed distance deployment of RIS has an issue in that {the} RIS location can go beyond the cell of its associated BS for smaller cells. It is more intuitive to specify an RIS location relative to the cell. In this section, we consider a case where RIS is deployed at a distance proportional to the cell radius. 
Since several steps are similar to the previous case, we have skipped the derivations. Note that cell radius is equal to half of the distance $\distHorzNN$ between the associated BS $X_0$ and the nearest neighbor $X_1$, hence,
we first calculate the distribution of $\distHorzNN$. 

\begin{lemma}
\label{lemma:condistr}
Conditioned on the horizontal distance $\distHorzBU$ between the typical  user and the serving BS $X_0$, the conditional probability distribution of the  distance $\distHorzNN$ of  the nearest neighbour BS $X_1$ of the serving BS \ie $f_{\distHorzNN\mid \distHorzBU}(y|x)$ is given by (See Appendix \ref{app:lemma_condistr} for proof)
\begin{align}
 &f_\mathrm{\distHorzNN \mid \distHorzBU}(y|x) = 
 \begin{cases}
    \lambda_\mathrm{b} \exp\left(-\lambda_\mathrm{b} y\right) &  \text{if } 0 \leq y < 2x, \\
    2\lambda_\mathrm{b} \exp\left(-2\lambda_\mathrm{b} (y - x)\right) &  \text{if } y \geq 2x.
    \end{cases}
 \label{eq:condistr}
\end{align}%
\end{lemma}
We deploy the RIS on both sides of BS each at a distance $\frac{f \distHorzNN}{2}$ from the BS, where $f \in (0, 1)$ is the design parameter. 

 \subsection{Characterization of Serving Link Blockages and Connection Failure Probability}
\label{sec:LBRIS}


Now, we  derive the joint blockage probability similar to the previous case.
\begin{lemma}
Conditioned on the horizontal distance $\distHorzBU$ between the typical  user and the BS, and the nearest neighbor distance $\distHorzNN$, the joint probability of the events $E_\mathrm{b}$ and $E_\mathrm{s}$ is 
\begin{align*}
\jpV{ \distHorzBU,\distHorzNN}= &\mathbb{P}(\E_\mathrm{b}\cap \E_\mathrm{s} \mid \distHorzBU, \distHorzNN) \\
= &\begin{cases}
    \left(1 - \exp\left(-\lambda_\mathrm{v}\frac{h_\mathrm{v}}{h_\mathrm{b}}\distHorzBU\right)\right)  
     \left(1 - \exp\left(- \lambda_\mathrm{v} (\frac{f \distHorzNN}{2} - \distHorzBU)\frac{h_\mathrm{v}}{h_\mathrm{s}}\right)\right)&\text{ if }  0 \leq \distHorzBU \leq \frac{f \distHorzNN}{2}, \\
    1 - \exp\left(- \lambda_\mathrm{v} (\distHorzBU - \frac{f \distHorzNN}{2})\frac{h_\mathrm{v}}{h_\mathrm{s}}\right)& \text{ if } \distHorzBU > \frac{f \distHorzNN}{2}.
    \end{cases}
\end{align*}
\end{lemma}
By unconditioning on $\distHorzBU$ and $\distHorzNN$, using \eqref{eq:dist_lane} and \eqref{eq:condistr}, and after simplifying, we get the joint blockage probability aka connection failure probability as
\begin{align}
    \jp{}=\mathbb{P}(\E_\mathrm{b}\cap \E_\mathrm{s}) =
    & \frac{(f - 1)(f - 2)}{2(4\rseverity + 1 - f)} + \frac{f^3}{2(4\bseverity + f)(4\rseverity + f)} + \frac{2\rseverity}{(4\rseverity + 1 - f)(2\rseverity + 1)}.
    \label{eq:joint_block2}
\end{align}
Note that for a fixed $f$, the joint blockage probability decreases as we increase $\rseverity$ as the ratio $\rseverity$ has a degree $-1$ in all the terms of the above expression. This means that for a fixed density of BSs and blockages, the joint blockage probability can be decreased if we increase the height of RISs or decrease the height of blockages. It can also be observed that the height of RIS is more influential in affecting the joint blockage probability than the height of the BS.

\subsection{SNR Coverage Probability}


\begin{lemma}
The coverage probability of the typical  user, when RISs are deployed at $f$ fraction of cell radius is
\begin{align}
    \pcu  &
    \approx \sum\limits_{n = 1}^{n_0}(-1)^{n+1}\binom{n_0}{n}\left[
    \expects{\distHorzBU}{
    \expS{
    {-\frac{n \gamma'  }{\ellBU(\distHorzBU)} 
    }
    }  
    (1-\jpD{\distHorzBU})}
    \nonumber \right. \\
    & \ \ \ + \left.\expects{\distHorzBU,\distHorzNN}{
    \expS{-\frac{n \gamma'  
    }
    {
    \ell_\mathrm{r}(\distHorzBU, f\distHorzNN/2)}} (\jpD{\distHorzBU}-\jpV{\distHorzBU,\distHorzNN})}\right]
\end{align}
where
$\gamma' =\eta_m \snrthreshold \noisen  /(u_mv_m)$.
%
Here, the expectation in the second term is taken with respect to the distributions of the distance between BS and user $\distHorzBU$, and the nearest neighbour distance of the BS $\distHorzNN$ given in  \eqref{eq:dist_lane} and \eqref{eq:condistr} respectively.
\end{lemma}
This expression can be further simplified as done in Corollary 1.

\subsection{SINR Coverage Probability}

We will now incorporate the effect of interference in our analysis to observe  the respective trends for the SINR coverage probability. When the associated BS is LOS (\ie event $\A_\bs$), the RISs are not involved, hence the analysis is exactly the same. 
Hence, we focus on the event $\A_\ris$ 
which is the case when the tagged BS is blocked while at the same time, the corresponding RIS is not blocked.


In this case, for a given $\frac{\distHorzNN}{2}$ and  $\distHorzBU$, the interfering BSs lie in the range $[\max{(r, \distHorzNN - \distHorzBU)}, \frac{d_2 h_{b}}{h_{v}}]$ on the left side. This range represents all unblocked BSs on the left side of the user. The lower limit is $\max{(r, \distHorzNN - r)}$, 
The lower limit is due to two restrictions, first  there can be no other BS at a distance less than $\distHorzBU$ and second as no other BS can be less than  $\distHorzNN$ distance from the associated BS. 
%
%
Similar to the fixed distance deployment case, given $\dR,\dL,\distHorzBU,\distHorzNN$, the interference is given as
\begin{align*}
    \intfV=& \sum\nolimits_{X_{i} \in \PhiL \backslash\left\{B_{0}\right\}} g_{i} \ellBU (\|X_i\|)U_i v_{ms}^{f} \cdot \mathbbm{1} \left(\|X_i\| < \fracS{\dL h_{\bs}}{h_{\vehicles}} \right) \cdot \mathbbm{1} \left(\|X_i\| > \max{(\distHorzBU, \distHorzNN - \distHorzBU)} \right) \\
     &+ \sum\nolimits_{X_{i} \in \PhiR} h_{i} \ell_{\mathrm{rI}}(X_i,\distHorzUR,f\distHorzNN/2)U_i v_m \cdot \mathbbm{1} \left(\distHorzBU+\distHorzNN<X_i< \distHorzBU+ z_\ris \right)
\end{align*}
where $v_{ms}^f = \left[v_m \cdot \mathbbm{1} \left(r < \frac{f \distHorzNN}{2}\right) + v_s \cdot \mathbbm{1} \left(r > \frac{f \distHorzNN}{2}\right) \right]$ and $z_\ris=\sqrt{\left(h_\ris - h_\bs\right)^{2}+{\left(\frac{f \distHorzNN}{2}\right)}^{2}} \cdot \theta_\beam$.
Now, the Laplace transform of the interference is given as
\begin{align}
&\mathcal{L}_{\intfV}(s \mid r, \distHorzNN, d_1, d_2) = \exp \left( -\lambda_{b} \left(\int_{\max(r, \distHorzNN - r)}^{\fracS{\dL h_{b}}{h_{v}}} 1 - \mathbb{E}_U\left[\frac{1}{1+ s \ellBU(x)U v_{ms}^f}\right]^{n_0} \dd x  \nonumber \right. \right. \\
&\left. \left. +  \indside{z_\ris>\distHorzNN}\int_{\distHorzBU+ \distHorzNN}^{\distHorzBU + z_\ris} 1 - \mathbb{E}_U\left[\frac{1}{1+ s \ell_{\mathrm{rI}}(x,\distHorzBU,f\distHorzNN/2)U v_m}\right]^{n_0} \dd x
 \right) \right)\nonumber\\
&\times \mathbb{E}_U\left[\frac{1}{1+ \indside{z_\ris\ge \distHorzNN} s \ell_{\mathrm{rI}}(\distHorzBU+ \distHorzNN,\distHorzBU,f\distHorzNN/2)U v_m}\right]^{n_0} 
\end{align}
where the last term is due to interference from the nearest BS.
Now the coverage probability can be computed as
\begin{small}
\begin{align*}
&\pcuI
=\sum_{n=1}^{n_0}
			(-1)^{n+1} \binom{n_0}{n} \bigg[\\\
			&
			\int_0^\infty\int_{\distHorzBU h_\vehicles/h_\bs}^\infty\int_0^\infty 
			e^{- \frac{ n \gamma \noisen \eta_m}{\ellBU(\distHorzBU)u_m v_m}  }
		\laplaces{I|\distHorzBU,\dR,\dL}{\textstyle \frac{n\gamma\eta_m  }{\ellBU(\distHorzBU)u_m v_m} }
 f_{\dR}(\dR)  f_{\dL}(\dL)f_{\distHorzBU}(\distHorzBU) \dd\dR\dd\dL\dd \distHorzBU\\
&+
\int_0^\infty 
\int_{2\distHorzBU/f}^\infty
\int_0^{\distHorzBU h_\vehicles/h_\bs}
\int_{(f\distHorzNN/2-\distHorzBU) h_\vehicles/h_\ris}^\infty 
e^{- \frac{ n \gamma \noisen \eta_m}{\ell_\mathrm{r}(\distHorzBU,f\distHorzNN/2)u_m v_m}  }
		\laplaces{\intfV|\distHorzBU,\dR,\dL}{\textstyle \frac{n\gamma\eta_m  }{\ell_\mathrm{r}(\distHorzBU,f\distHorzNN/2)u_m v_m} }
		\\&
\qquad\qquad  f_{\dR}(\dR)  f_{\dL}(\dL)f_{\distHorzNN}(\distHorzNN)f_{\distHorzBU}(\distHorzBU)\dd\dR\dd\dL\dd  \distHorzNN \dd \distHorzBU
\end{align*}\begin{align*}
&+
\int_0^\infty 
\int_{0}^{2\distHorzBU/f}
\int_{(\distHorzBU-f\distHorzNN/2) h_\vehicles/h_\ris}^{\distHorzBU h_\vehicles/h_\bs}\int_{0}^\infty 
e^{- \frac{ n \gamma \noisen \eta_m}{\ell_\mathrm{r}(\distHorzBU,f\distHorzNN/2)u_m v_m}  }
		\laplaces{\intfV|\distHorzBU,\dR,\dL}{\textstyle \frac{n\gamma\eta_m  }{\ell_\mathrm{r}(\distHorzBU,f\distHorzNN/2)u_m v_m} }\\
		&
		\qquad \qquad
f_{\dR}(\dR)  f_{\dL}(\dL)f_{\distHorzNN}(\distHorzNN)f_{\distHorzBU}(\distHorzBU) \dd\dR\dd\dL \dd  \distHorzNN\dd \distHorzBU
\bigg].%
\end{align*}%
\end{small}%
%
%

Till now, we have considered a typical user at a street and hence, have limited our focus to a single street. We note that in a scenario with multiple streets, it may be useful to deploy additional RISs on these intersections as it enables RIS to serve two streets. In the next section, we study this deployment. 

\section{Impact of Intersection Mounting of RISs}
\vspace{-0.15cm}
In an urban setting, intersections of streets are crucial. We now consider an additional set of RISs that are mounted at each street intersection and are associated with a nearby BS  on each crossing street.
The density of the intersection-mounted RISs in the network is governed by $\lambda_\mathrm{r}$.

\newcommand{\general}{}

If the direct link and the associated RIS link are blocked to the typical \general user, it attempts to connect to the nearest BS $Y_0$ from the neighboring street via the nearest intersection-mounted RIS on the opposite side as compared to the serving BS. Let the distance to the nearest intersection from the typical general user on the opposite side as the serving BS be given by $r_{\ue\xsec}$. The PDF of $r_{\ue\xsec}$ is~\cite{ghatak2019small}\vspace{-1mm}
\begin{align}
    f_{r_{\ue\xsec}}(r) = 2\lambda_\mathrm{r} \exp(-2 \lambda_\mathrm{r} r).
    \label{eq:dist_intersection}
\end{align}
Thus, the distance to the nearest intersection-mounted RIS is $d_{\ue\xsec} = \left(r_{\ue\xsec}^2 + h_\ris^2\right)^{1/2}$.
The distance between the nearest intersection-mounted RIS and the nearest BS $Y_0$ on the adjoining street is $d_{\xsec\bs} = \left((h_\mathrm{b} - h_\mathrm{s})^2 + r_{\xsec\bs}^2\right)^{1/2}$, where $r_{\xsec\bs}$ is the horizontal distance of the adjoining-street nearest BS from the intersection. Note that $r_{\xsec\bs}$ is independent and identically distributed as $\distHorzBU$ and its PDF is given by \eqref{eq:dist_lane}.
%
Let $\A_\xsec$ denote the event the user is able to establish an \ac{LOS} connection to the nearest intersection-mounted RIS and the link to the nearest BS as well as its associated RIS is blocked. 
Therefore, 
$$\A_\xsec=\E_\bs\cap \E_\ris\cap \overline{\E_\xsec},$$
where 
$\overline{\E_\xsec}$ denotes the event that there exists an LOS path from the user to the nearest intersection-mounted RIS.
\begin{lemma}
Given $r_{\ue\xsec}$ and $\distHorzBU$, the association probability to the 
nearest intersection-mounted RIS for the typical user is 
    \vspace{-0.1cm}
\begin{align}
\prob{\A_\xsec|\distHorzBU, r_{\ue\xsec}}
= \begin{cases}
    \exp\left(-\lambda_\mathrm{v} \frac{h_\mathrm{v}}{h_\mathrm{s}} r_{\ue\xsec}\right) \left(1 - \exp\left(-\lambda_\mathrm{v}\frac{h_\mathrm{v}}{h_\mathrm{b}}
    \distHorzBU
    \right)\right)  \\ 
    \times \left(1 - \exp\left(- \lambda_\mathrm{v} (r_\mathrm{s}-\distHorzBU)\frac{h_\mathrm{v}}{h_\mathrm{s}}\right)\right) \ \ \ & \text{ if } \   0 \leq \distHorzBU {<} r_\mathrm{s} \\
    \exp\left(-\lambda_\mathrm{v} \frac{h_\mathrm{v}}{h_\mathrm{s}} r_{\ue\xsec}\right) \left(1 - \exp\left(- \lambda_\mathrm{v} (\distHorzBU-r_\mathrm{s})\frac{h_\mathrm{v}}{h_\mathrm{s}}\right)\right) 
& \text{ if } \ \distHorzBU  {\geq} r_\mathrm{s},
    \end{cases}.
\end{align}
\end{lemma}
\begin{IEEEproof}
The proof follows directly from the independence of the events $\Bar{\E}_\xsec$, $\E_\mathrm{s}$, and $\E_\mathrm{b}$ for a given $r_{\ue\xsec}$ and $\distHorzBU$. 
\end{IEEEproof}
Consequently, the coverage probability of the typical \general user can be written as follows:

\begin{theorem}
The SNR coverage probability of the {typical \general user} when it is aided by an intersection-mounted RIS in case of joint blockage of the serving BS and its associated RIS is
\begin{align}
  \mathsf{p}_{\ue\xsec}(\gamma) \approx &
  \pcu
    +
    \sum\nolimits_{n = 1}^{n_0}(-1)^{n+1} {\binom{n_0}{n}}\mathbb{E}
    \left[ 
    \exp\left(-\frac{{n \gamma' 
    }}
    {\ell_\mathrm{r}(r_{\ue\xsec}, r_{\xsec\bs})
}\right)
    \mathbb{P}(\A_\xsec
    | \distHorzBU, r_{\ue\xsec})
    \right].\nonumber
        \vspace{-0.7cm}
\end{align}
Here, the expectation is taken with respect to the independent RVs $\distHorzBU$, $r_{\ue\xsec}$, and $r_{\xsec\bs}$. 
\end{theorem}

We now extend our interference analysis in the previous section to the intersection mounted RIS case. As the SINR for $\A_\bs$ and $\A_\ris$ case remains the same, we focus on $A_\xsec$ here. 
$A_\xsec{}$ occurs when the intersection mounted RIS is LOS, while both the user's BS and the corresponding RIS are blocked. This is equivalent to the following conditions 
\begin{align*}
A_\xsec&=A_{\xsec1}\cup A_{\xsec2} \text{ with}\\
A_{\xsec1}&=\left\{\distHorzBU>r_\ris, \dR < \fracS{(\distHorzBU ) h_{\vehicles}}{h_{\bs}},
\dR < \fracS{(\distHorzBU - r_\ris) h_{\vehicles}}{h_{\ris}},\dL > \fracS{(r_{\ue\xsec}) h_{\vehicles}}{h_{\bs}}\right\}\\
A_{\xsec2}&=\left\{\distHorzBU<r_\ris, \dR < \fracS{(\distHorzBU ) h_{\vehicles}}{h_{\bs}},
 \fracS{(r_\ris-\distHorzBU) h_{\vehicles}}{h_{\ris}}>\dL > \fracS{(r_{\ue\xsec}) h_{\vehicles}}{h_{\bs}}\right\}.
\end{align*}

Now, the interference at the user will only be due to the BSs on the left side of the user under $A_{\xsec1}$ and no BS interference under $A_{\xsec2}$. Moreover, there will also be interference due to the BSs on the adjacent street that are interfering via the intersection mounted RIS. 
The via RIS interference power from each BS at $x$ distance from intersection at adjacent street on the same side as the new serving BS is given as
\begin{align}
    P_\rx'(x)  = 
    h_i  { U_i  \ellBR(x) \pi (N - 1)^2 } 
    { \ellUR(r_{\ue\xsec}) v_m}
    =h_i \ell_\mathrm{xrI}(x,r_{\ue\xsec})U_i v_m,
\end{align}
where $\ell_\mathrm{xrI}(x,r_{\ue\xsec})=\ellBR(x) \pi (N - 1)^2  
    { \ellUR(r_{\ue\xsec})}$.
Hence, under $\A_{\xsec1}$, the interference is
%
%
\begin{align*}
    I'' &= \sum_{X_{i} \in \PhiL} g_{i} \ell_\mathrm{b}(\|X_i\|)U_i v_m \cdot \mathbbm{1} \left(\|X_i\| < \fracS{\dL h_{\bs}}{h_{\vehicles}} \right)
     &+ \sum_{X_{i} \in \Phi_\mathsf{y}\setminus\{Y_0\}} h_{i} \ell_{\mathrm{xrI}}(\|X_i\|) U_i v_m \cdot \mathbbm{1} \left(\|X_i\| < r_{\xsec\bs}+ z_\ris \right),
\end{align*}
where $\Phi_\mathsf{y}$ denotes the PPP modeling the locations of BS located in the adjacent street on the same side as new serving BS is located and  $z_\ris=\sqrt{\left(h_\ris - h_\bs\right)^{2}+r_{\xsec\bs}^{2}} \cdot \theta_{\beam}$.
Here, since the tagged BS is blocked, and the user is pointing at the opposite direction of the tagged BS, then its main lobe will align with all the interfering BSs on that particular side of the user. 
Now, the interference's Laplace transform is given as
%
%
\begin{align}
&\mathcal{L}_{I''}(s | \distHorzBU, r_{\ue\xsec}, \dR, \A_{\xsec1}) = \exp \left( -\lambda_{\bs} \left(\int_{\distHorzBU}^{\infty} 1 - \mathbb{E}_U\left[\frac{1}{\left(1+ s U v_m \ell_\mathrm{b}(x) \indside{x < \fracS{\dL h_{\bs}}{h_{\vehicles}}}\right)^{n_0}}\right] \dd x \right)\right)  \nonumber \\
&\times\exp \left( -\lambda_{b} \left(\int_{r_{\xsec\bs}}^{r_{\xsec\bs}+z_\ris} 1 - \mathbb{E}_U\left[\frac{1}{\left(1 + s U v_m \ell_{\mathrm{rxI}} (x)  \right)^{n_0}}\right] \dd x \right)\right).\label{eq:LTIAx1}
\end{align}

Similarly, under $A_{\xsec2}$, the interference is given by just the second term in \eqref{eq:LTIAx1}
Now, given $r_{\ue\xsec}$ and $r_{\xsec\bs}$, the coverage probability when assisted by intersection RIS is given as
%
%
%
\begin{align*}
&\mathrm{p_{ic}}(\gamma \mid r_{\ue\xsec}, r_{\xsec\bs}) = \int_{r_{s}}^{\infty} \bigg[ 
\prob{d_1<\fracS{(\distHorzBU - r_\ris) h_{\vehicles}}{h_{\ris}}} 
\times   \\
& \left . \int_{\frac{r_{\ue\xsec} h_{\vehicles}}{h_{\ris}}}^{\infty}  \exp\left(-\frac{{ \gamma \noisen 
    }}{
    \ell_\mathrm{r}(r_{\ue\xsec}, r_{\xsec\bs})  {u_m v_m}
    }\right) \mathcal{L}_{I''| \distHorzBU, r_{\ue\xsec}, \dR, \A_{\xsec1}} \left(\frac{{ \gamma
    }}{
    \ell_\mathrm{r}(r_{\ue\xsec}, r_{\xsec\bs})  {u_m v_m}
    }
    \right) f_{\dL}(\dL) \cdot \mathrm{d} \dL \right] \cdot f_{\distHorzBU}(\distHorzBU) \dd \distHorzBU\\
  &  +\int_0^{r_{s}} \bigg[ 
\prob{d_1<\fracS{(\distHorzBU ) h_{\vehicles}}{h_{\bs}}} \int_{\fracS{r_{\ue\xsec} h_{\vehicles}}{h_{\ris}}}^{\fracS{({r_\ris - \distHorzBU}) h_{\vehicles}}{h_{\ris}}}  \exp\left(-\frac{{ \gamma \noisen 
    }}{
    \ell_\mathrm{r}(r_{\ue\xsec}, r_{\xsec\bs}  {u_m v_m})
    }\right)
\times   \\
& \left . \hspace{1.5in} \mathcal{L}_{I''|\distHorzBU, r_{\ue\xsec}, \dR, \A_{\xsec2}} \left(\frac{{ \gamma
    }}{
    \ell_\mathrm{r}(r_{\ue\xsec}, r_{\xsec\bs}) {u_m v_m}
    }
    \right) f_{\dL}(\dL) \cdot \mathrm{d} \dL \right] \cdot f_{\distHorzBU}(\distHorzBU) \dd \distHorzBU.
\end{align*}
Here, the term containing $\dR$ is written separately because other terms in the integral are independent of $\dR$. Hence, the total coverage probability is
\begin{align*}
   \mathsf{p}_{\ue\xsec}(\gamma) = \pcuI{}+ {\mathbb{E}_{r_{\ue\xsec}, r_{\xsec\bs}}} \left [ \mathrm{p_{ic}}(\gamma | r_{\ue\xsec}, r_{\xsec\bs}) \right].
\end{align*}



\subsection{Typical Intersection User} 
 We now consider an additional type of user known as  the typical intersection user located at the  street intersections which are \ac{PLT} crossings~\cite{jeyaraj2020cox}. This is in contrast with the typical general user that can be located anywhere on the street.
The typical general user connects to the nearest BS on the same street, whereas, the typical intersection user connects to the nearest BS on either of the streets that comprise the intersection. Thus, for the typical intersection user, the PDF of the horizontal distance to its serving BS is given by\vspace{-1mm}
\begin{align}
    f_{\distHorzBIU}(x) = 4\lambda_\mathrm{b}\exp(- 4\lambda_\mathrm{b} x), x\ge 0. \label{eq:interserve}
\end{align}
Here, \eqref{eq:interserve} follows from the fact that, for an intersection user, there are four possible directions to the serving BS.
Similar to a general user, a BS serves an intersection user by the direct link if the link is in the \ac{LOS} state and from an associated RIS if the direct link is blocked (i.e., it is in the \ac{NLOS} state). 

We consider that the typical intersection user receives services only from the associated BS either from its corresponding RIS or the intersection-mounted RIS. If both the direct link and the associated RIS link are blocked to the typical intersection user, the BS attempts to serve the typical intersection user via the intersection-mounted RIS. Note that the distance between the typical intersection user and the base of the intersection-mounted RIS is zero ($r_{\iue\xsec} = 0$) since this RIS is located directly overhead the user. Thus, for the typical intersection user, the link to the nearest intersection-mounted RIS is always in \ac{LOS}. Thus, for the typical intersection user, we have intersection mounted RIS conditional association probability $\prob{\A_\xsec|\distHorzBIU}=\prob{\,\overline{\E_\xsec}\cap\E_\mathrm{s}\cap\E_\mathrm{b}| \distHorzBIU} = \prob{\E_\mathrm{s}\cap\E_\mathrm{b}| \distHorzBIU}=1-\prob{\A_\bs|\distHorzBIU}-\prob{\A_\ris|\distHorzBIU}$. 

\begin{theorem}
The SNR coverage probability of the {typical intersection user} aided by an intersection-mounted RIS in case of the joint blockage of the serving BS and associated RIS is
\begin{align}
    \mathsf{p}_{\iue\xsec} (\gamma)&\approx 
    \sum\limits_{n = 1}^{n_0}(-1)^{n+1} {\binom{n_0}{n}}
    \mathbb{E}\left[
    \vphantom{\frac{1}1}
    \exp\left(-\frac{n \gamma' 
    }{\ell_\mathrm{b}(\distHorzBIU)}\right)  \nonumber 
    \exp\left(-\lambda_\mathrm{v}\distHorzBIU \frac{h_\mathrm{v}}{h_\mathrm{b}}\right) 
    \nonumber \right. \\
    & +\!\left. \exp\left(-\frac{{n \gamma' 
    }}{ \ell_\mathrm{r}(\distHorzIUR, \distHorzBR)
    }\right) \mathbb{P}(\bar{\E}_\mathrm{s},\E_\mathrm{b}\mid \distHorzBIU) 
    +
    \!\exp\left(-\frac{n \gamma' }
    {\ell_\mathrm{r}(0,r_\mathrm{ib})
}
\right)
     \mathbb{P}(\E_\mathrm{s},\E_\mathrm{b}\!\mid\! \distHorzBIU
     )
     \right],\nonumber
\end{align}
where 
the expectation is taken with respect to $\distHorzBIU$ as given in \eqref{eq:interserve}.
\end{theorem}

The interference and SINR coverage probability is similar to the typical general user and hence skipped for brevity. 


\section{Numerical Results and Discussions}
\label{sec:NRD}

In this section, we present some numerical results to highlight the impact of several system parameters on the system performance. 
For the evaluation purpose, the transmit power of a BS is {$20$ dBm,  the total beam-forming gain $u_mv_m$ is $4$} and side-lobe gain  $u_mv_m$ is $0.77$. We consider the mm-wave band of $f_c = 28$ GHz with  1 GHz bandwidth. The propagation parameters follow the 3GPP specifications as defined in~\cite{38.900}, where the standard path loss is considered with $K = \left(\frac{c}{4\pi f_c}\right)^2$, the path-loss exponent 
$\alpha = 2$, and {$n_0 = 1$}. Here, $c$ is the speed of light in free space. The number of elements in each RIS is $N = 100$. The noise power spectral density is $-174$ dBm/Hz {with noise figure of $10$}. Unless mentioned explicitly, we assume $\lambda_\bs=0.05~/$m, $h_\bs=10$~m, $h_\mathrm{r}=50$~m, $\lambda_\mathrm{v}= 0.1~/$m and $h_\mathrm{v}= 3$~m.


\begin{figure}
    \centering
     {\small \bf (a)}\includegraphics[width = \figwidthSbS]{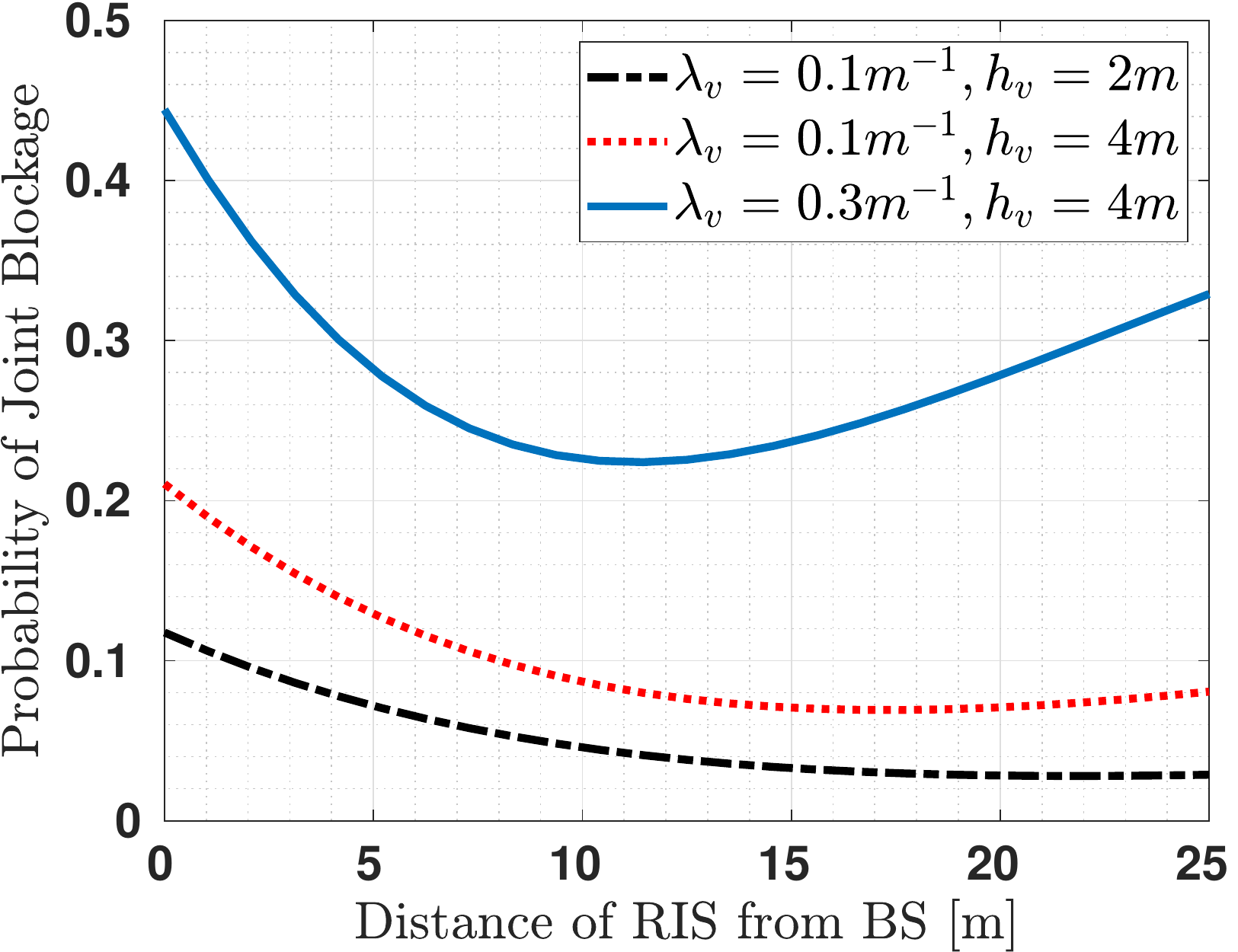}
      {\small \bf (b)}
     \includegraphics[width = \figwidthSbS]{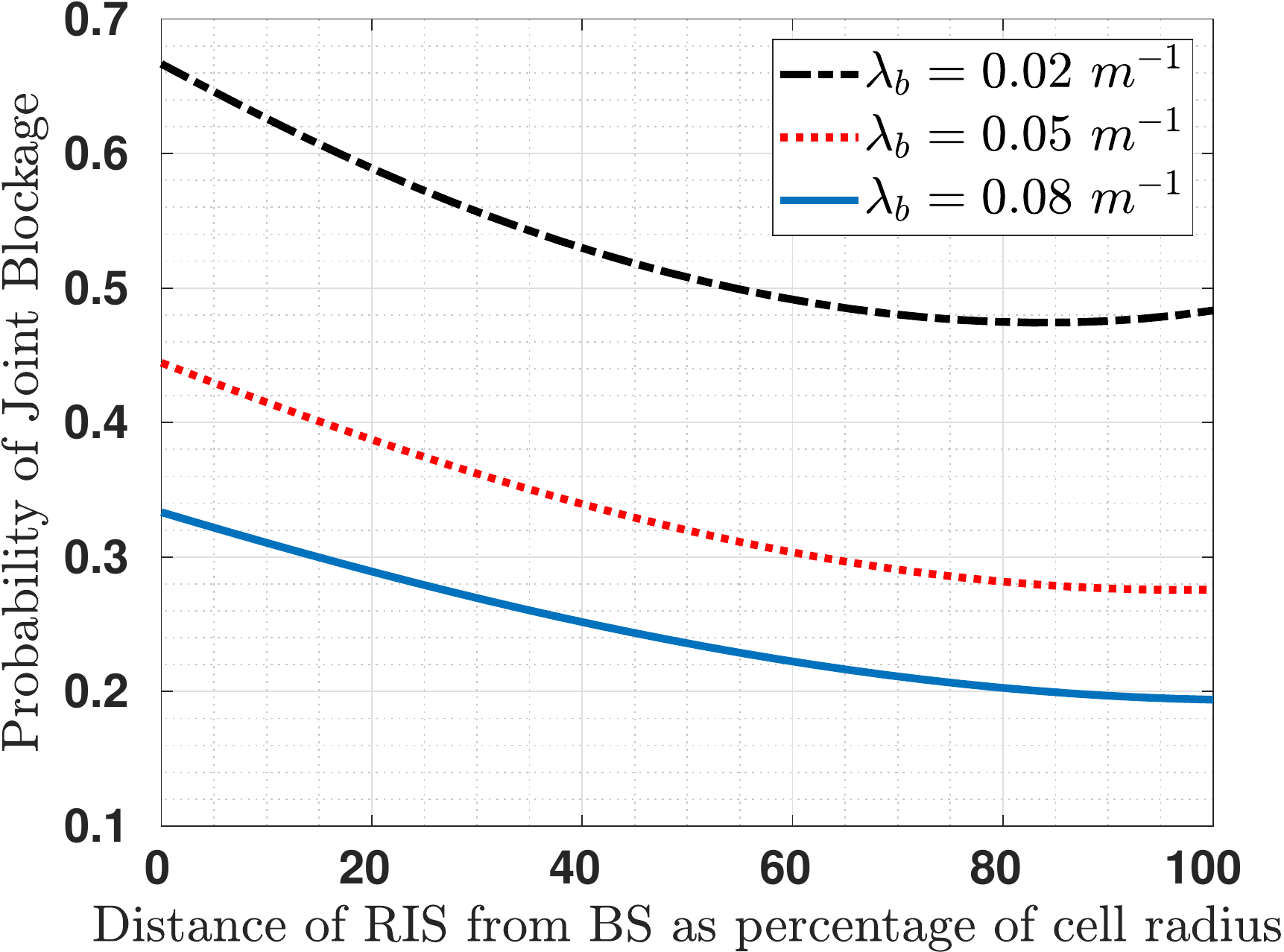}\\
    \caption{Connection failure probability or the joint blockage probability $\jp$ with respect to the deployment distance of the RIS for different values of blockage density and height for (a) fix distance deployment (b) cell radius dependent deployment at $f$ fraction of cell radius. Here $h_\mathrm{b} = 10$ m, {$h_\mathrm{s} = 15$ m, and $\lambda_\mathrm{b} = 0.05$ m$^{-1}$.} The optimum distance decreases with an increase in blockage severity with respect to the BS density.}
    \label{fig:Blockage1}
    \figspaceadjust
\end{figure}


\textbf{Blockage probability.} 
Fig.~\ref{fig:Blockage1}(a) shows that the probability that the links from the typical general user to its serving BS and the associated RISs are jointly blocked, is minimized for a certain value $\drisopt$ of $r_\mathrm{s}$ (i.e., the distance  of the RIS from its coupled BS).  Let $R=1/\lambda_\bs$ denote the order of average cell radius. As seen from Remark 1, if the RISs are deployed very far from most users (i.e., $r_\ris\gg R$) or very close to BSs (i.e. $r_\ris\approx 0$), $\jp{}$ is high.  This optimal value of $r_\mathrm{s}$ is of the order of $R$ and depends on the blockage density $\lambda_\mathrm{v}$ and the blockage height $h_\mathrm{v}$. When blockages are sparse (i.e.,  $\lambda_\mathrm{v}$ and $h_\mathrm{v}$ are small), cell edge users (i.e., $\distHorzBU\approx R$) still have a chance to be blocked. Hence, they determine the optimal value of $r_\ris$ to be equal to their location, i.e., $\drisopt\approx R$. On the other hand, in the presence of severe blockages, even the users  near BS (i.e., $\distHorzBU\approx R$) can be blocked, hence the $\drisopt$ shifts towards $R/2$.
One can also observe that the lower bound (given by Proposition \ref{pro:block_bound}) is tight for small values of $r_\mathrm{s}$ since it is obtained by considering only the case when $\distHorzBU \geq r_\mathrm{s}$.

Fig.~\ref{fig:Blockage1}(b) 
shows that the joint blockage probability when an RIS is deployed at $f$ fraction of the cell radius of the selected BS.  We observe that it is optimum to deploy RIS at 100\% of the cell radius for high BS density. But,  similar to fixed distance deployment,  decreasing $\lambda_\mathrm{b}$ reduces the optimal fraction $f$ to deploy the RIS. 
This is consistent with the discussion of \eqref{eq:joint_block2}. 



\begin{figure}
    \centering
     {\small \bf (a)}\includegraphics[width = \figwidthSbS]{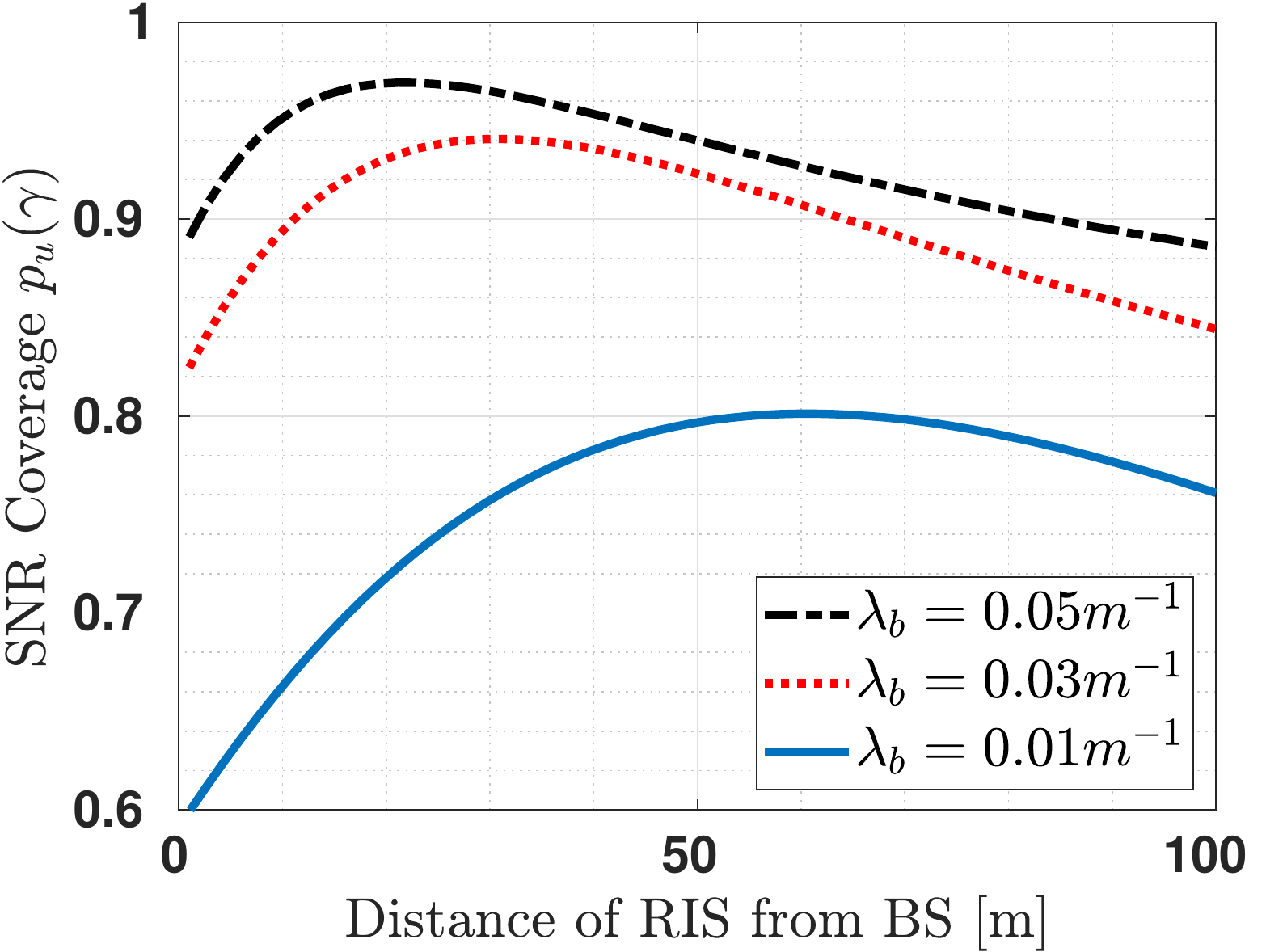}
       {\small \bf (b)}\includegraphics[width = \figwidthSbS]{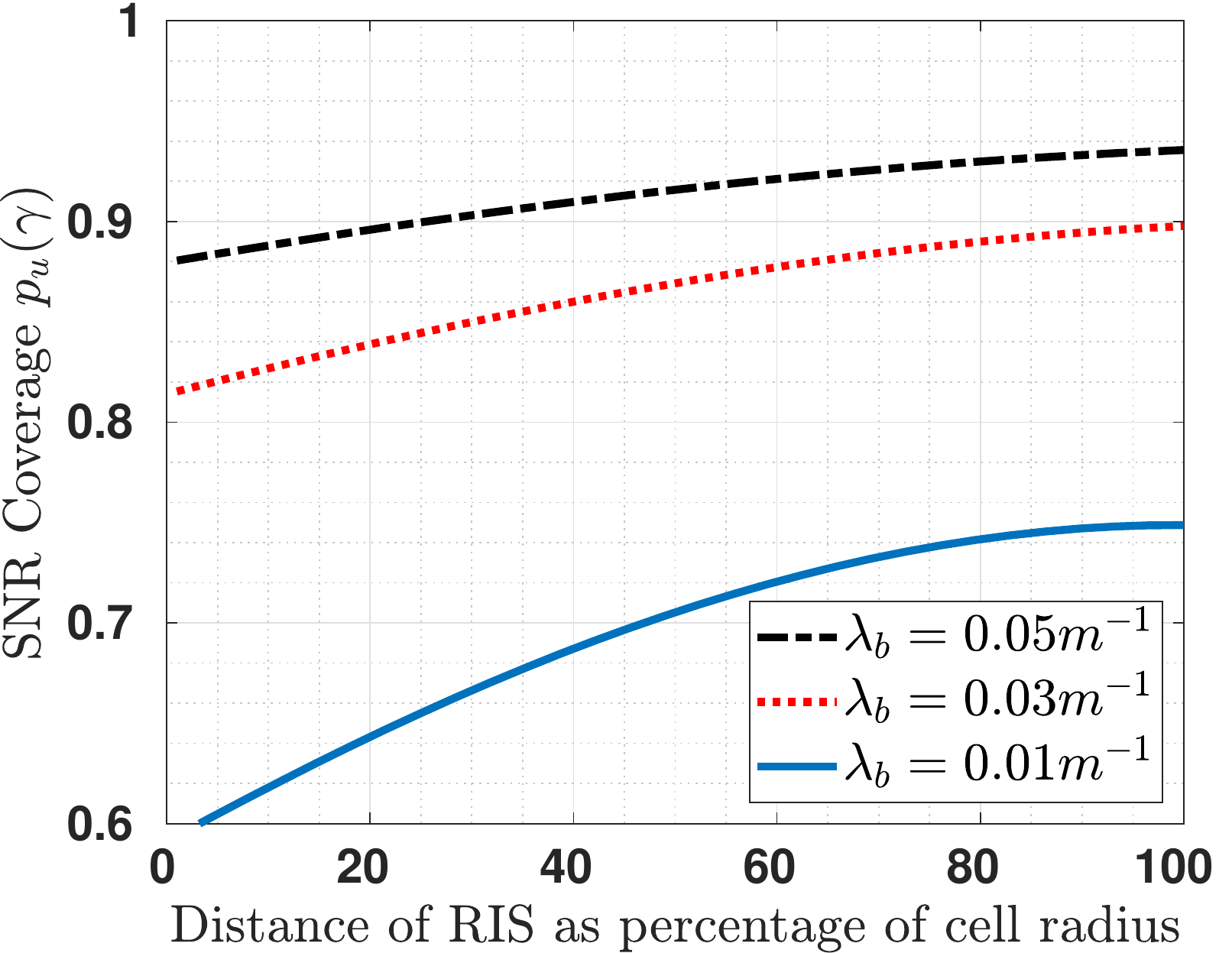}\\
    \caption{Variation of SNR coverage probability $\pcu$ with the deployment distance of the RIS for different values of the BS density for (a) fixed distance and (b) cell size dependent deployment with $f$ fraction of cell radius. Here, $h_\mathrm{b} = 10$ m, $h_\mathrm{s} = 50$ m, and $\lambda_\mathrm{v} = 0.5$ m$^{-1}$.}
    \label{fig:SNR_rs}
\figspaceadjust
\end{figure}


\textbf{SNR coverage probability.}  
Fig.~\ref{fig:SNR_rs}(a) shows that the SNR coverage probability also attains a maximum value for a certain optimal value $\drisoptb$ of $r_\ris$. As the deployment of the BSs gets denser, the optimal distance of the RIS to BS decreases  and remains of the order of $R$. For example, for $\lambda_\mathrm{b} = 0.05$~m$^{-1}$, $R=20$~m and $\drisoptb= 12.5$~m, whereas it decreases to 5~m for $\lambda_\mathrm{b} = 0.02$ m$^{-1}$ ($R=5$~m). Due to the high probability of joint blockage (as shown in Fig.~\ref{fig:Blockage1}), the coverage probability decreases for sufficiently low and high values of $r_\mathrm{s}$. 

 Fig.~
\ref{fig:SNR_rs}(b) shows that 
 when BS density is higher or comparable to the blockage density, the coverage probability 
achieves its maximum when each RIS is deployed at the edge of the corresponding cell (\ie $f$=100\%). 
This is intuitive because, with an increase in $f$, the joint blockage probability decreases and hence, BSs are able to cover more users. However, when the blockage density is  high as compared to the BS density, we observe the optimum to occur when the RIS is deployed at nearly 70\% of the cell radius. Intuitively, it can be seen that in the case when BS density is low, more users will experience signal blockage not only from BS, but also from RIS if RIS is deployed too far. Therefore, it is advisable to deploy RIS towards the middle of the cell than the edge 
 to improve coverage probability to users.

\begin{figure}[t!]
    \centering
    {\small \bf (a)} \includegraphics[width = \figwidthSbS]{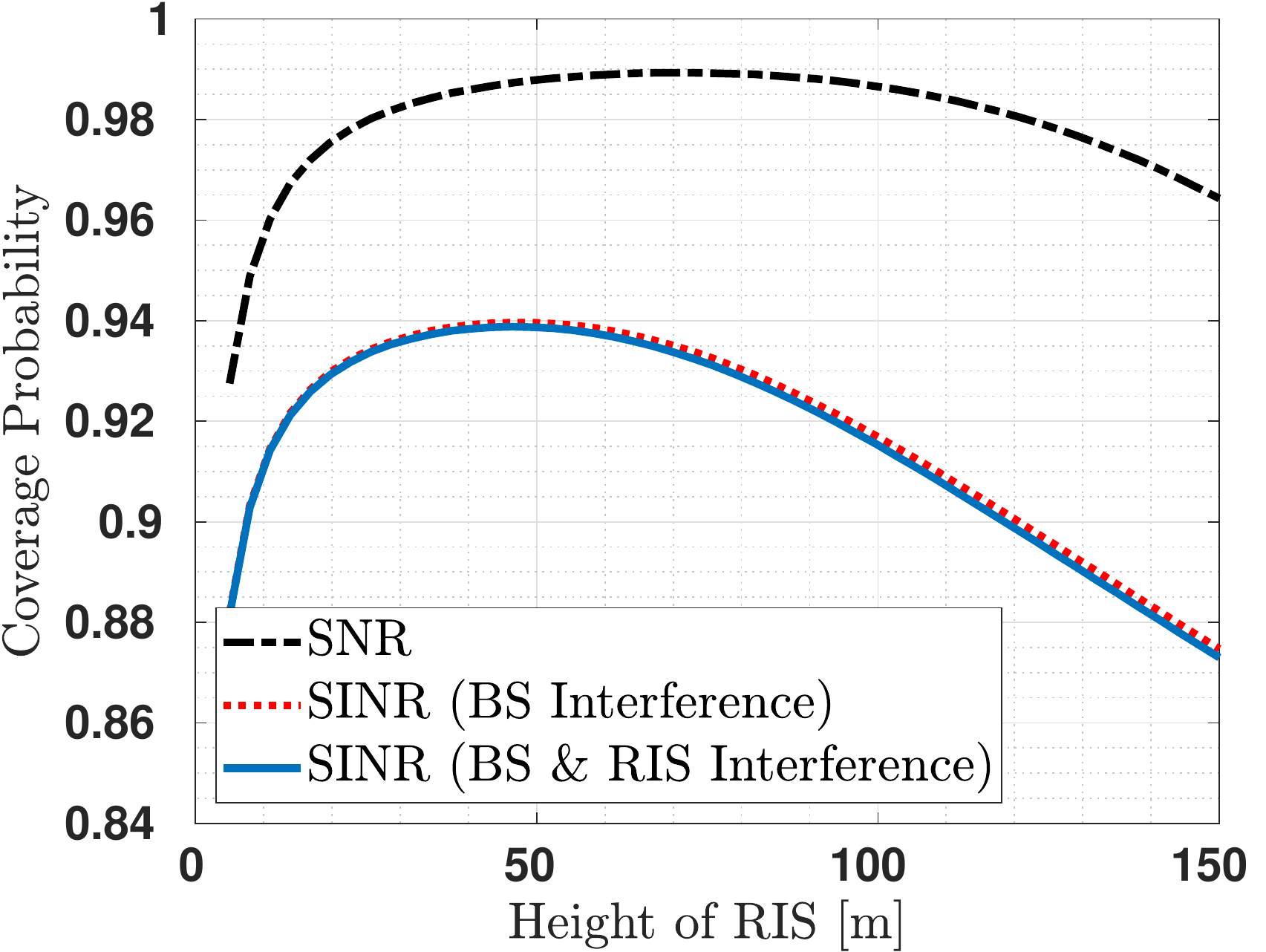}
     {\small \bf (b)}
      \includegraphics[width = \figwidthSbS]{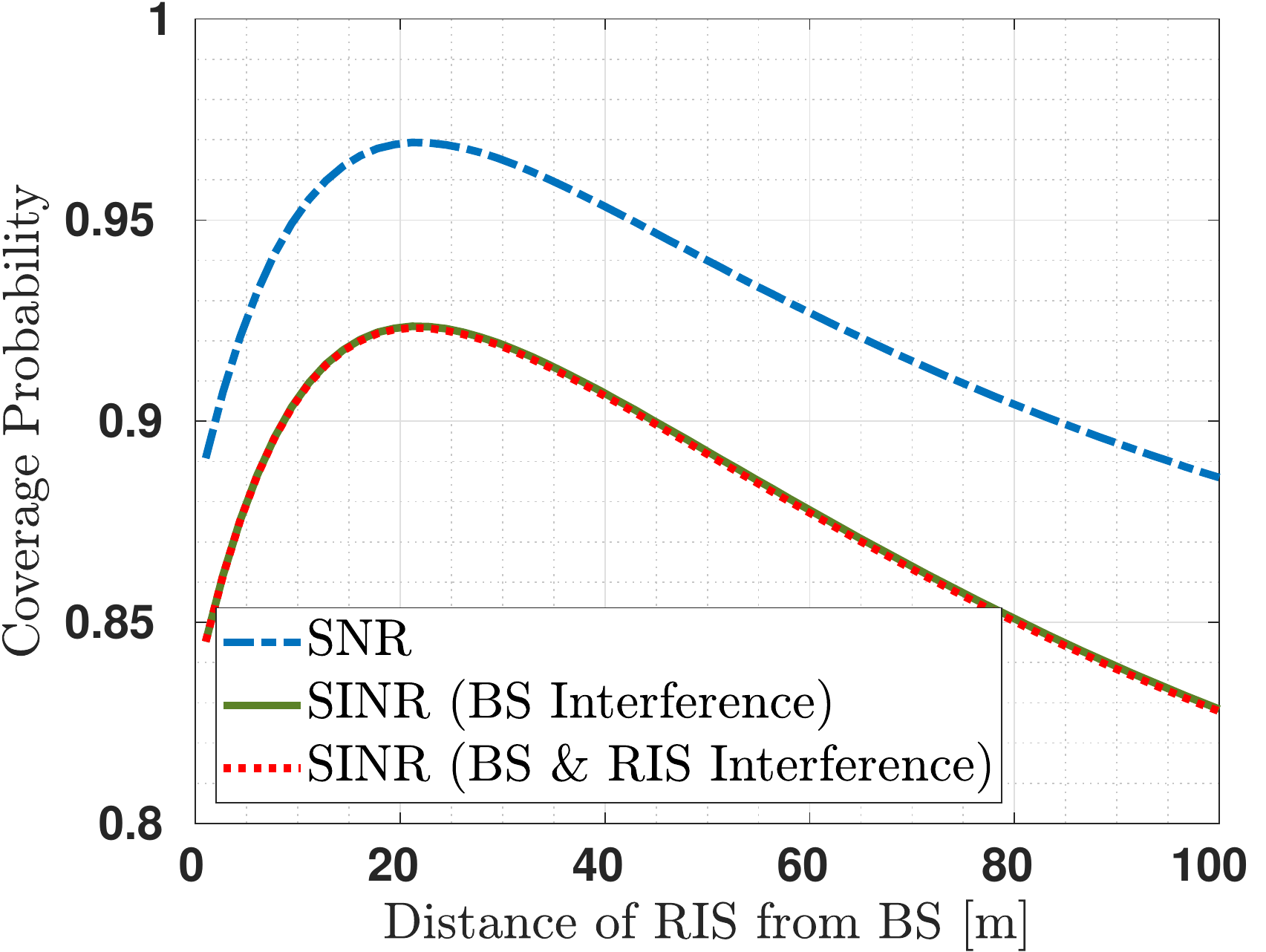}\\

    \caption{Variation of SNR and SINR coverage probabilities with respect to the deployment height $h_\ris$ {and location $r_s$} of the RIS for various values of blockage height and density. Here, $\lambda_{\rm b}$ = 0.05 m$^{-1}$, $h_b = 10 m$ and {$r_s = 21.59$ m.}}
    \label{fig:SNR_height}
\end{figure}

Interestingly, Fig.~\ref{fig:SNR_height}(a) shows that for a BS height of 10 m, the optimal height  $h_{\rm s}$ for deploying the RIS is about 50 m, which is not affected significantly by $h_\vehicles$ and $\lambda_\vehicles$. 
A further increase in $h_{\rm s}$ reduces the coverage probability $\pcu$. For very large values of $h_\mathrm{s}$,  $\pcu$ becomes constant because the coverage is only due to the direct link to the BS owing to the large distance between the RIS and the user.

\textbf{SINR coverage and effect of interference from {via-RIS} links.}
{Fig. \ref{fig:SNR_height}(b)} shows the SINR coverage probability for the fixed distance deployment case. It also shows when 
the interference $\intfV_2$ from via-RIS  links is ignored. 
We can see that the deployment location of the RIS shifts slightly closer to the BS location when via-RIS interference is included. However, other than that, the interference from the selected RIS has practically no impact on the deployment location of the RIS with respect to the BS and hence, can be ignored.
{It can also be observed from Fig. \ref{fig:SNR_height}(b) that the optimal deployment location of RIS is approximately the same whether we calculate it in terms of SINR or SNR. The same is also seen in Fig. \ref{fig:SNR_rsbsinterf} which shows the variation of SINR coverage probability with $r_\ris$ for various values of the BS and blockage density.  It can be seen that SNR and SINR curves are closely aligned with respect to each other when plotted with respect to $r_s$ and that the optimal deployment location is always close to $1/\lambda_b$.}

\begin{figure}[ht!]
\centering
      {\small \bf (a)}\includegraphics[width = \figwidthSbS]{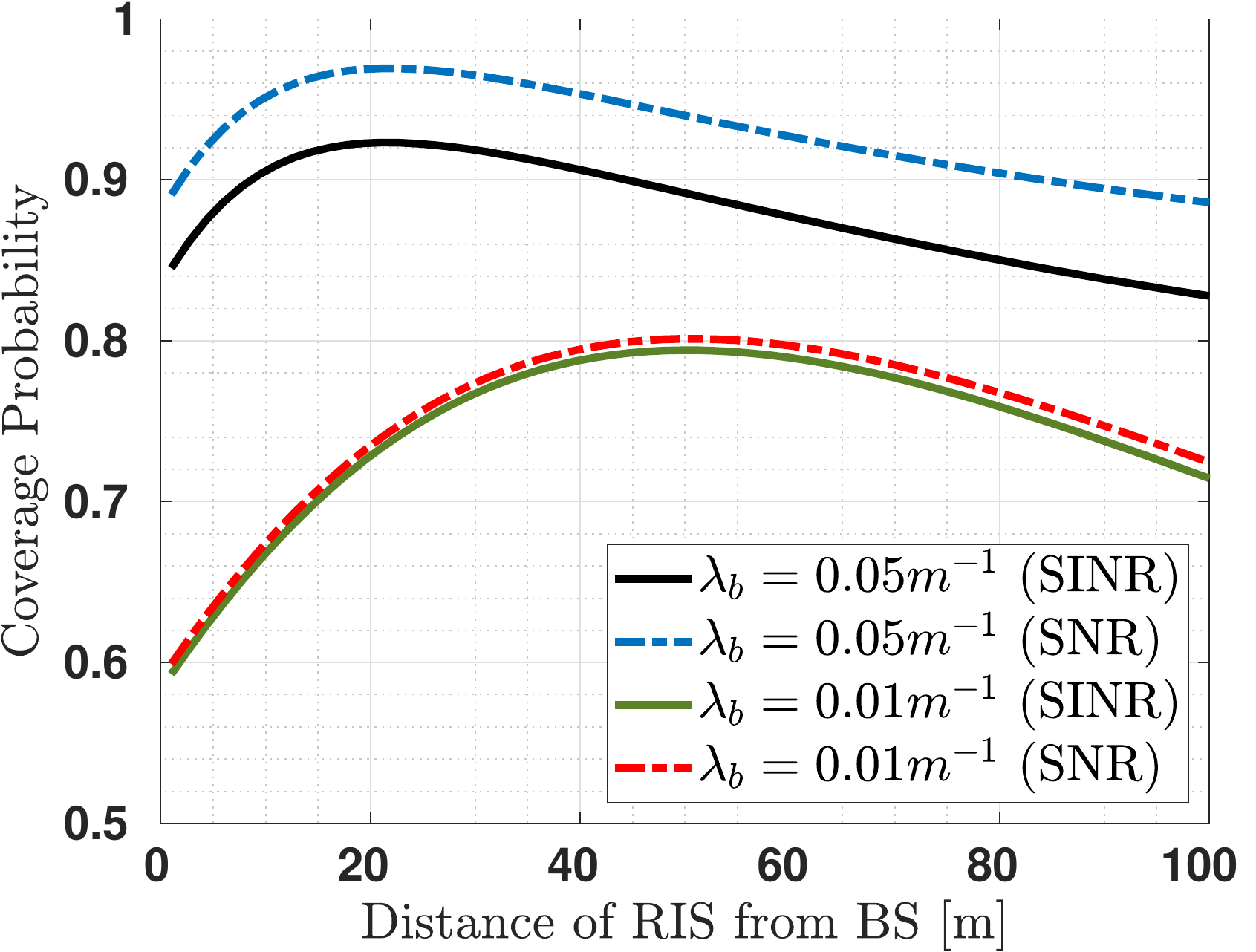}\ 
      {\small \bf (b)}\includegraphics[width = \figwidthSbS]{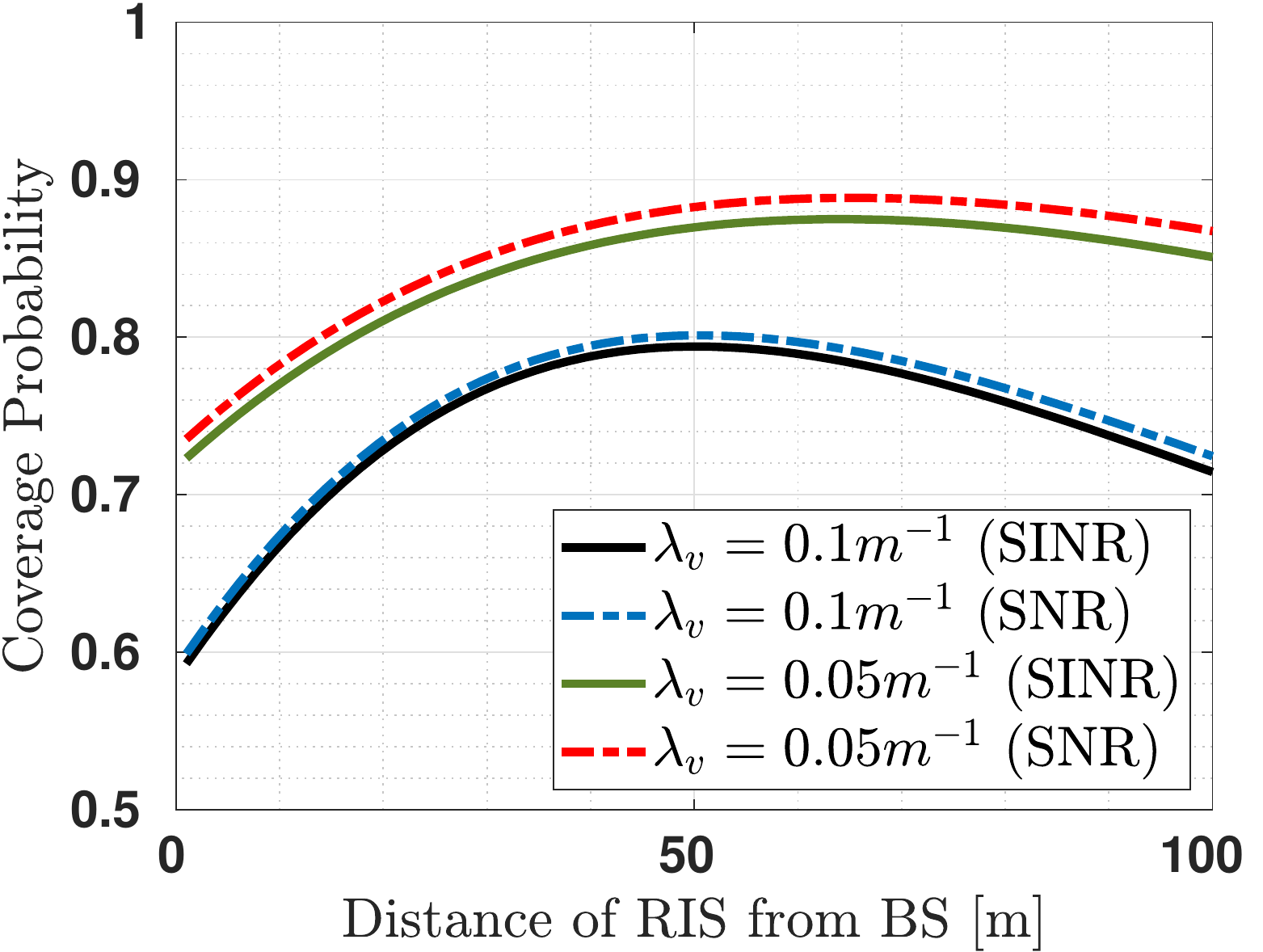}
    \caption{Variation of SINR coverage probability $\pcuI$ with the deployment distance of the RIS 
for the fixed distance deployment for (a) Variation in BS Density (b) Variation in Blockage Density (Here $\lambda_b = 0.05 m^{-1}$). The RISs are deployed at an optimal location with respect to the SINR.
}
    \label{fig:SNR_rsbsinterf}
\end{figure}

\begin{figure}[ht!]
\centering
     {\small \bf (a)} \includegraphics[width = \figwidthSbS]{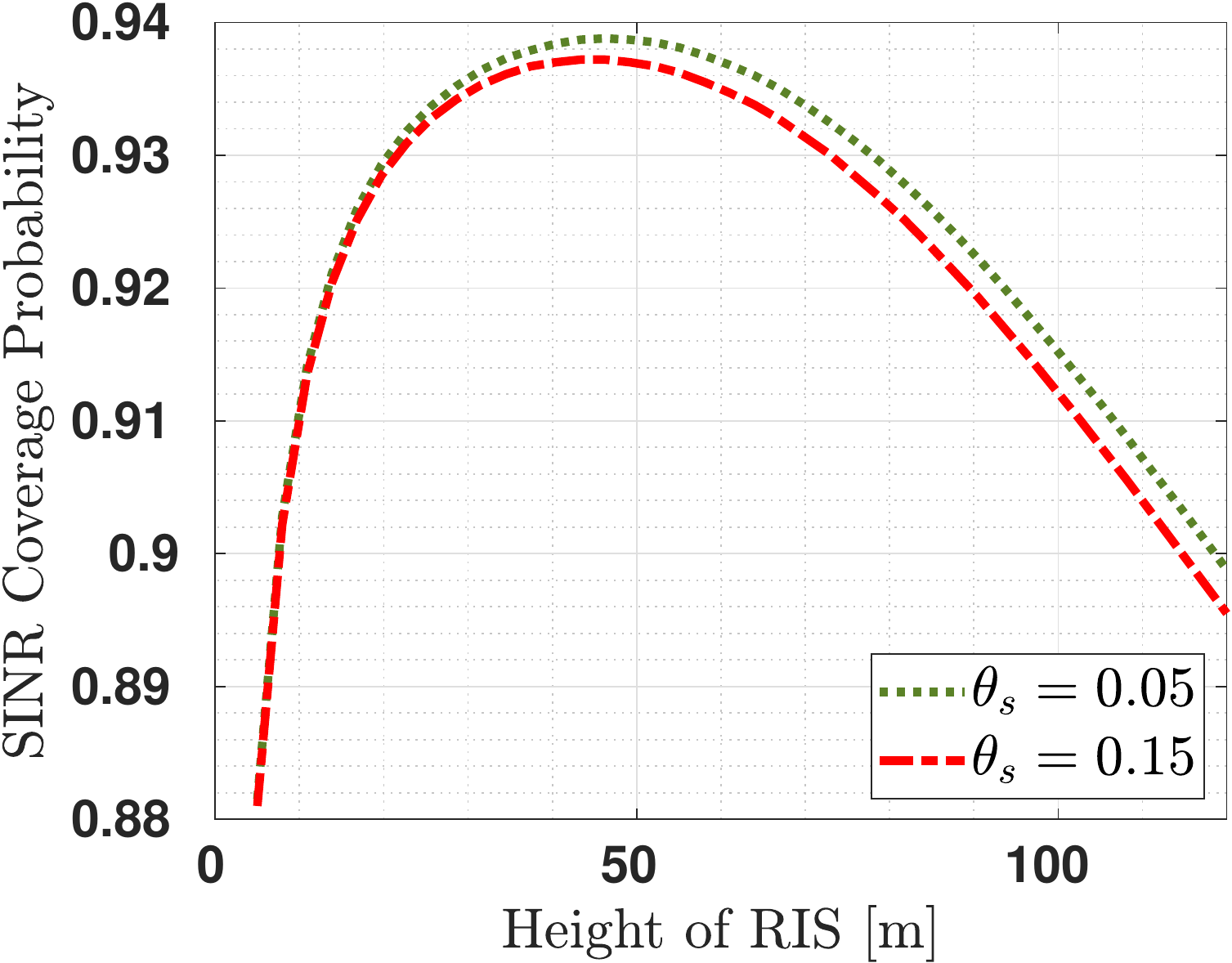}\ 
     {\small \bf (b)} \includegraphics[width = \figwidthSbS]{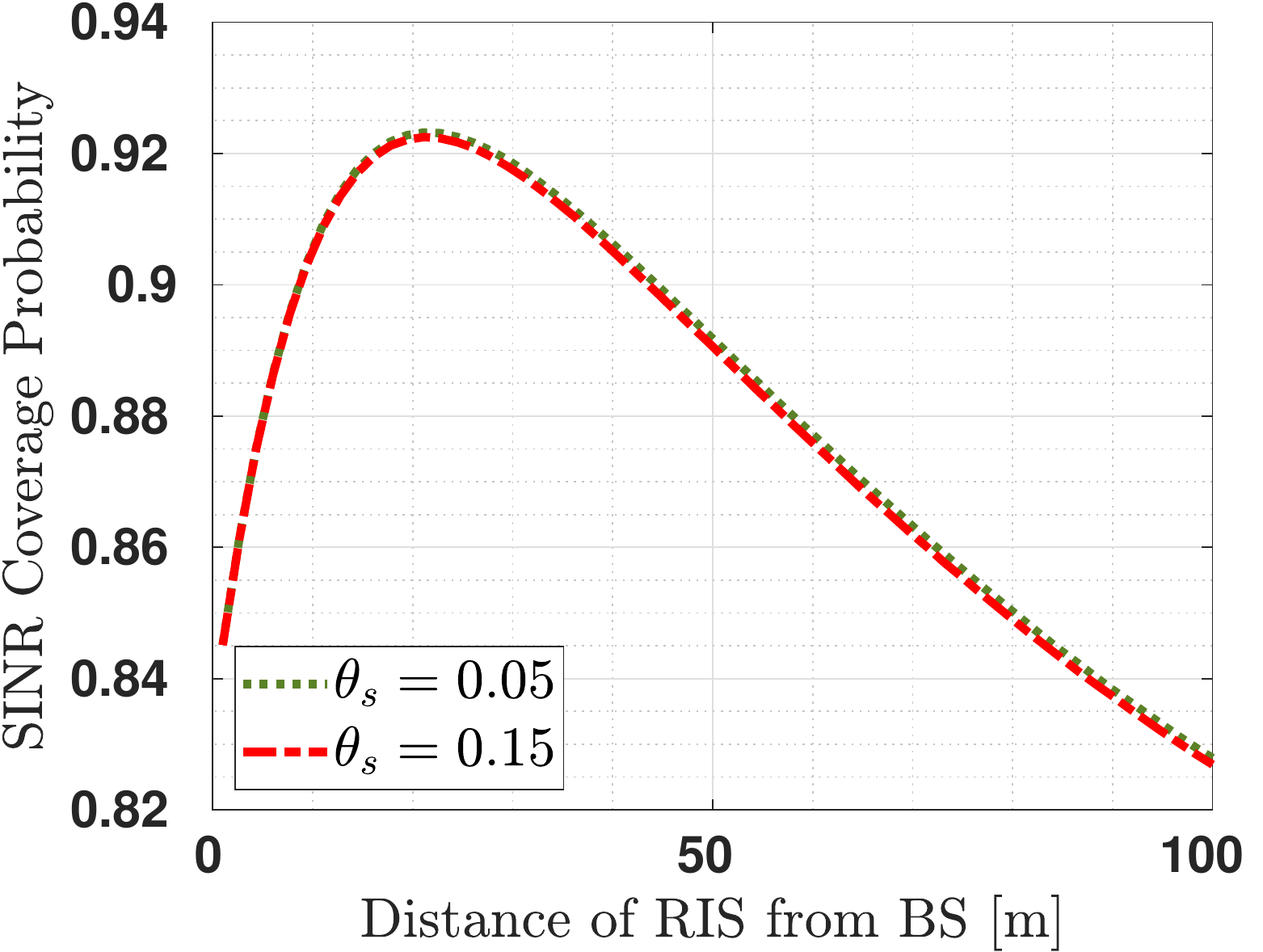}
    \caption{Variation of SINR coverage probability $\pcuI$ for two different values of $\theta_b$ 
for the fixed distance deployment for (a) Variation in Height of RIS (b) Variation in RIS Deployment Location.
}
    \label{fig:SNR_rsbsinterf2}
\end{figure}

{
\textbf{Impact of RIS beamwidth angle.}
Fig. \ref{fig:SNR_rsbsinterf2} shows the effect of the beamwidth $\theta_b$ of RIS on the SINR coverage probability. As expected, it can be observed that $\theta_b$ has very little effect on the coverage probabilities when plotted with respect to RIS height or deployment distance. 
}

\begin{figure}
    \centering
     {\small \bf (a)} \includegraphics[width = .45\linewidth,trim=30 20 70 40,clip]{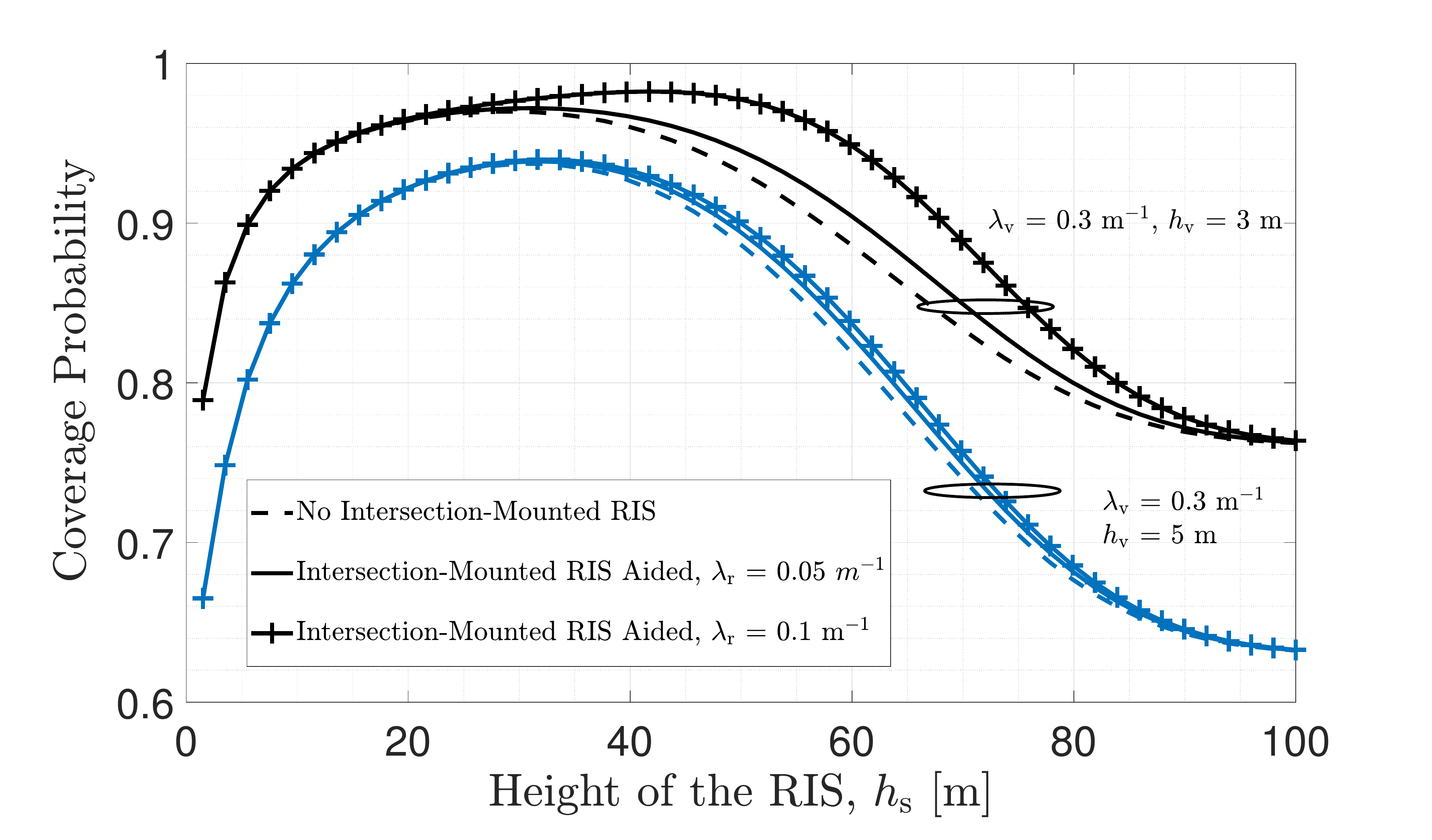}
     {\small \bf (b)} \includegraphics[width = .45\linewidth,trim=30 20 70 40,clip]{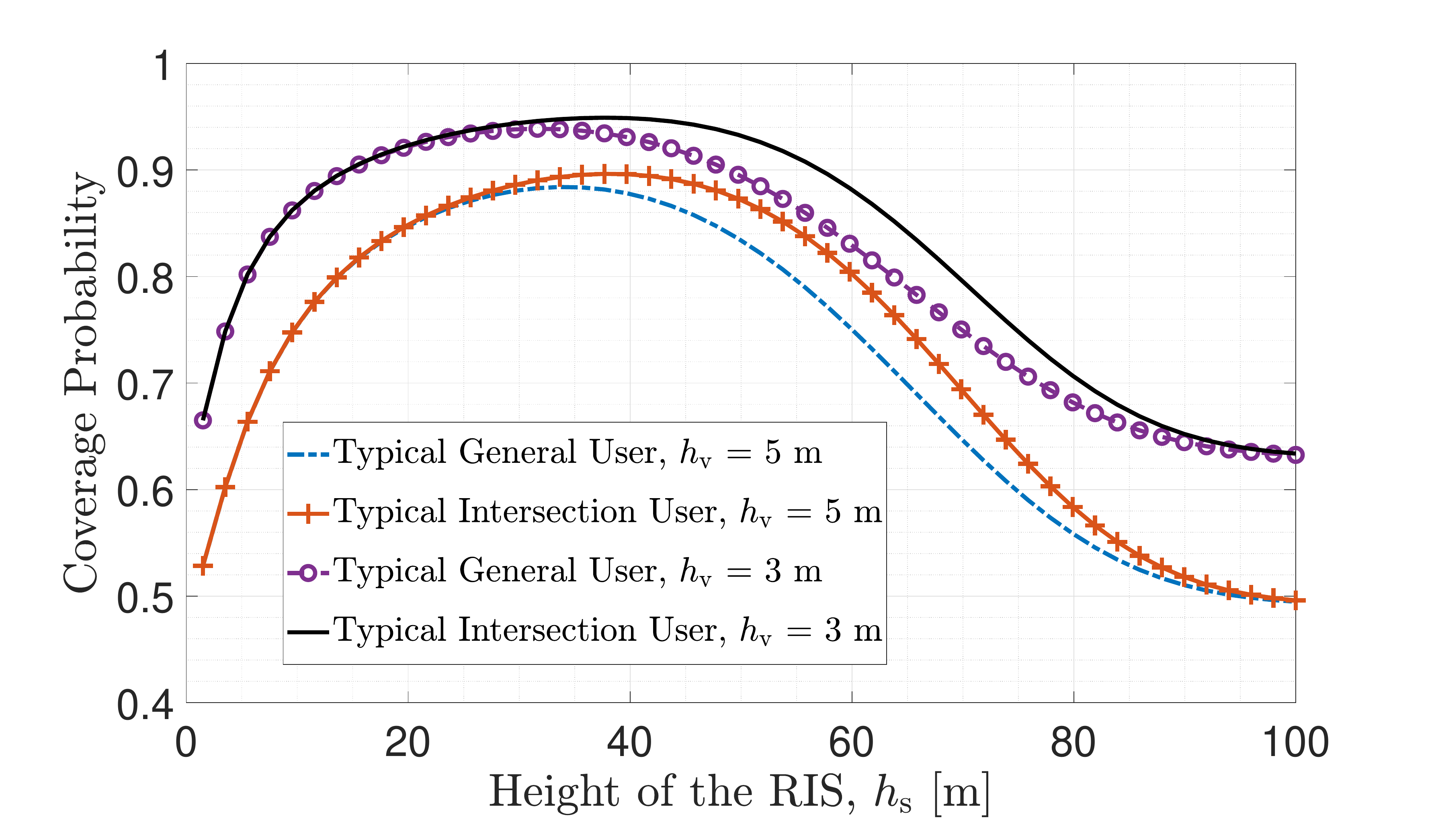}
    \caption{Variation of the SNR coverage probability with respect to the deployment height of the RIS. (a) Impact of the intersection-mounted RIS for typical general user and (b) Comparison between the typical general and intersection user for $\lambda_{\rm r}$ = 0.1 m$^{-1}$. }
    \label{fig:xmounted}
    \figspaceadjust
 \figspaceadjust
\end{figure}

\textbf{Impact of intersection-mounted {RIS}.}
Fig.~\ref{fig:xmounted}(a) shows the variation of SNR coverage probability of  the typical general user with RIS height in presence of intersection-mounted RISs. We can observe that leveraging an intersection-mounted RIS when the serving BS and its associated RIS are jointly blocked  increases the coverage probability. We see that the optimal height of RISs also increases in this case. We can observe that the gain obtained from the use of intersection RISs increases with road density which can be justified as below. For a dense road network (i.e., high $\lambda_{\rm r}$), the intersection-mounted RIS is near to the user, which enables a higher RIS deployment while still maintaining the user coverage. Also, with a higher density $\lambda_{\rm r}$ of the streets, the number of intersection-mounted-RISs increases, and hence the coverage is larger. Fig.~\ref
{fig:xmounted}(b) shows that since the link between the intersection user and an intersection-mounted RIS is always in the \ac{LOS} state, such a user benefits from a higher mounting of the RIS.

\begin{figure}[ht!]
    \centering
    \includegraphics[width = \figwidth,trim=10 10 10 40,clip]{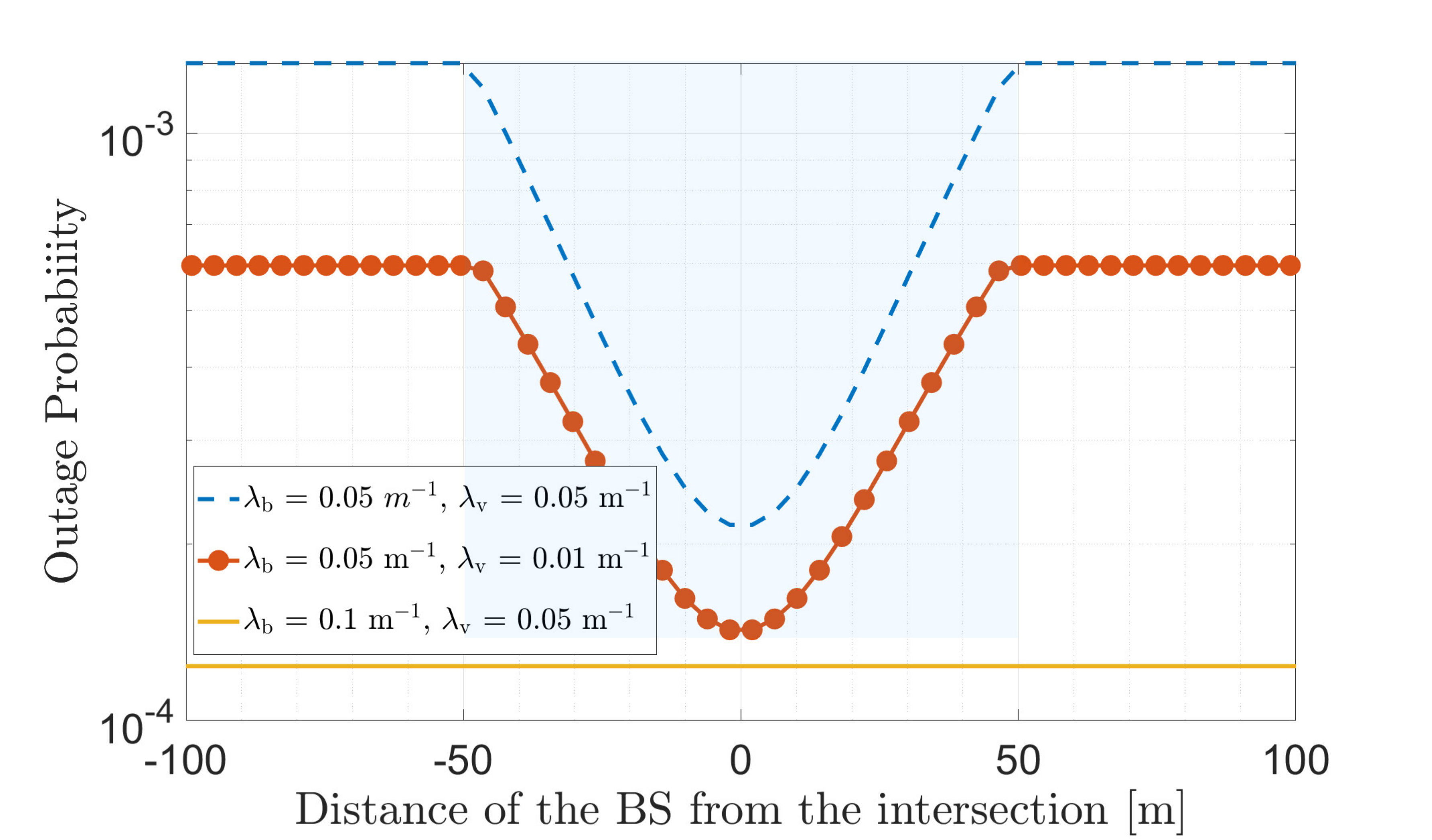}
    \caption{Outage probability of a user that selects either the intersection RIS or the associated RIS based on the higher received power. $\lambda_{\rm r}$ = 0.1 m$^{-1}$.}
    \label{fig:comp}
    \vspace{-0.1cm}
\end{figure}

\textbf{Comparison between association and intersection RISs.} 
For users associated to a BS that is close to an intersection, the intersection-mounted RIS might provide a higher signal power via a BS on the adjoining street as compared to the RIS coupled with the serving BS. In such cases, it might be worth bypassing the strongest BS association.
In Fig.~\ref{fig:comp}, we plot the outage probability of a tagged user uniformly located in the coverage area of a BS parameterized on the distance between the serving BS and the nearest intersection. For sparse deployments (i.e., low $\lambda_{\rm b}$), the distance, $d_{\rm bi}$, between the associated BS and the nearest intersection from the UE is large. Accordingly, the associated BS and the corresponding RIS provide statistically higher power as compared to the intersection RIS. On the contrary, when $d_{\rm bi}$ decreases, the intersection RIS may deliver a higher downlink power thereby increasing the coverage probability. Furthermore, recall that for low values of
$\lambda_{\rm b}$, the user can be far from the serving BS thereby
making the intersecting RIS a better choice for service. This results in a reduced outage of the users, as shown in the blue shaded region of~\ref{fig:comp}. The minimum outage is observed
when the associated BS is exactly located at the nearest intersection,
i.e., $d_{\rm bi} = 0$. However, for dense deployments (i.e., high $\lambda_{\rm b}$), the user is statistically closer to the serving BS and thus the associated BS and RIS association provides a higher downlink power as compared to the intersection mounted RIS. Consequently, the intersection mounted RIS association does not decrease the outage as seen from the figure for the case of $\lambda_{\rm b} = 0.1$ m$^{-1}$.

\section{Conclusions}

\label{sec:Con}
In this paper, we studied the optimal placement of RISs from the perspective of coverage. {Deploying RISs very close or very far from the BS increases the probability that the links from the user to the BS and to the RIS are simultaneously blocked. This results in worse coverage for users.  We showed that there is an optimal location inside the coverage region that improves the coverage for users on average. When RISs are deployed according to the cell radius unless the BS density is very low compared to the blockage density, RISs should be deployed at 100\% of the cell radius, and the coverage probability is maximized and the joint blockage probability is minimized at this value. However, for a low BS density, RISs should be deployed at nearly 70\% to 80\% of the cell radius depending upon the blockage density. We also observed that leveraging intersection-mounted RISs improves coverage. Also, the typical intersection user enjoys a higher coverage probability compared to the typical general user due to the existence of an LOS link from the former to the intersection-mounted RIS. Interestingly, for users associated to BSs that are deployed very close to an intersection, selecting an intersection-mounted RIS provides a higher coverage compared to selecting an RIS associated with the serving BS.}


\appendices

\section{Proof of Lemma \ref{lemma:optrsblock}}\label{app:optrsblock}

%

From the first optimality condition,
\begin{align*}
    \frac{\dd \jp{}}{\dd r_\ris}&=0 &
 \implies   \bseverityinv(\rseverityinv+2)e^{-(\rseverityinv-2) \lambda_\bs r_\ris}&=(\rseverityinv^2-4)e^{-a_1 \lambda_{\vehicles} r_\ris}+4(2+\bseverityinv-\rseverityinv).
\end{align*}
Now let 
$x=e^{-\bseverityinv \lambda_\bs r_\ris}$ to get
\begin{align*}
	x^{(\rseverityinv-2)/\bseverityinv}&=\frac{\rseverityinv-2}{\bseverityinv} x+4\frac{2+\bseverityinv-\rseverityinv}{  \bseverityinv(\rseverityinv+2)}
\end{align*}
which is the desired result.
The approximate solution can be found by approximating $e^{-y} \approx 1-y+y^2/2$ in the above equation.


\section{Proof of Theorem \ref{thm1}}\label{app:thm1}
The probability of the joint event that the user observes an \ac{LOS} link to the RIS (i.e., $\bar{\E}_\mathrm{s}$) and that the link to the BS is blocked (i.e., $\E_\mathrm{b}$), given that the user-BS horizontal distance is $r_\mathrm{bu}$ is given as 
$
     \mathbb{P}\left(\bar{\E}_\mathrm{s}, \E_\mathrm{b}\mid \distHorzBU\right) 
     = \mathbb{P}\left(\E_\mathrm{b}\mid  \distHorzBU \right) - \mathbb{P}\left({\E}_\mathrm{s}, \E_\mathrm{b}\mid  \distHorzBU\right)$. Therefore, 
     the coverage probability is
\begin{align}
    \pcu= \mathbb{E}_{\distHorzBU} \left[\mathbb{P}(S \geq \gamma\mid \distHorzBU )\ \mathbb{P}(\Bar{\E}_B\mid \distHorzBU) + \mathbb{P}(S' \geq \gamma \mid \distHorzBU)\  \mathbb{P}(\bar{\E}_\mathrm{s},\E_\mathrm{b}\mid \distHorzBU) \right] \nonumber.
    \end{align}
Now, substituting the values of SNRs $S$ and $S'$, we get
\begin{align}
 \pcu     = &  \mathbb{E}\left[\mathbb{P}\left(
     \frac{ \ell_\mathrm{b}(\distHorzBU) u_m v_m}{\noisen }  g_{\rm{bu}} 
      \geq \gamma \mid  \distHorzBU \right) \times \exp\left(-\lambda_\mathrm{v}\distHorzBU{h_\mathrm{v}}/{h_\mathrm{b}}\right)   \nonumber \right. \\
     &\left. + \mathbb{P}(\bar{\E}_\mathrm{s},\E_\mathrm{b}\mid \distHorzBU)  \times \mathbb{P}\left(
    \frac{ \ell_\mathrm{r}(\distHorzBU, r_{\rm s})u_m v_m}{\noisen }
    g_{\rm{su}} \geq \gamma \mid  \distHorzBU \right)\right] \nonumber .
\end{align}
Averaging over the fading random variables $g_{\rm{bu}}$ and $g_{\rm{su}}$, we get the result. 

\section{Proof of Lemma \ref{lem:intfD}} 
\label{app:lap}

From \eqref{eq:IRFD},  the Laplace transform can be written as,
\begin{align*}
& \mathcal{L}_{I(\PhiR)}(s) =  \mathbb{E}_{\PhiR}\left[\prod_{X_{i} \in \Phi \backslash \{B_{0}\}} \mathbb{E}_{g_{i}, U}\left[ \exp \left(- g_{i} \ell_\mathrm{b}(\|X_i\|)U_i v_m \cdot \mathbbm{1} \left(\|X_i\| < \frac{\dR h_{b}}{h_{v}} \right) \right) \right]\right],
\end{align*}
as $g_is$ are independent. Since $g_i$'s are Gamma distributed due to Nakagami fading, we know that
\begin{equation}
 \mathbb{E}_{g_{i}} \left[ \exp \left(- t g_i \right) \right] = \left(\frac{1}{1 + t}\right)^{n_0}.
\end{equation}
Hence, using the PGFL of a PPP, and after removing the indicator function by changing the limits of integration, we get 
\begin{align}
\mathcal{L}_{I(\PhiR)}(s) = \exp \left( -\lambda_\bs \left(\int_{r}^{\frac{\dR h_{\bs}}{h_{\vehicles}}} 1 - \expects{U}{\left(\frac{1}{1+ s \ell_\mathrm{b}(x)U v_m}\right)^{n_0}} \dd x \right) \right)
\end{align}
Now,  $U$ is a uniform \ac{RV} which takes values $u_m$ and $u_s$ with equal probability. Hence, we get
\begin{align}
\mathcal{L}_{I(\PhiR)}(s )&= 
\exp \left( -{\lambda_\bs}
 \left(\int_{r}^{\max \left(r, \frac{\dR h_\bs}{h_{\vehicles}}\right)}
 1- \frac12\left(\frac{1}{1+ s \ell(x)u_m v_m}\right)^{n_0} -\frac12\left(\frac{1}{1+ s \ell(x)u_s v_m}\right)^{n_0} \dd x \nonumber \right) \right).
\end{align}
Similarly, we can obtain the expression for $\mathcal{L}_{I(\PhiL)}(s)$ by replacing $v_m$ with $v_s$ since the user antennas are pointing away from these BSs. 

\section{Proof of Theorem \ref{thm:PcIFixed}} \label{app:PcIFixed}
The coverage probability is given as
\begin{align*}
\pcuI&=\expects{\distHorzBU,\dR,\dL}{\prob{\SINR>\gamma|\distHorzBU,\dR,\dL}}\\
&\stackrel{(a)}=\expects{\distHorzBU,\dR,\dL}{\prob{\SINR>\gamma,\A_\bs|\distHorzBU,\dR,\dL}}+\expects{\distHorzBU,\dR,\dL}{\prob{\SINR>\tau,\A_\ris|\distHorzBU,\dR,\dL}}\\&+\expects{\distHorzBU,\dR,\dL}{\prob{\SINR>\gamma,\overline{\A_\bs\cap\A_\ris}|\distHorzBU,\dR,\dL}}\\
&\stackrel{(b)}=\int_0^\infty\int_{\distHorzBU h_\vehicles/h_\bs}^\infty\int_0^\infty \prob{\SINR>\gamma|\distHorzBU,\dR,\dL} f_{\dR}(\dR)  f_{\dL}(\dL)f_{\distHorzBU}(\distHorzBU) \dd\dR\dd\dL\dd \distHorzBU\\
&+\int_{0}^{r_\ris}{\int_{(r_\ris-\distHorzBU) h_\vehicles/h_\ris}^\infty\int_0^{\distHorzBU h_\vehicles/h_\bs}} \prob{\SINR>\gamma|\distHorzBU,\dR,\dL} f_{\dR}(\dR)  f_{\dL}(\dL)f_{\distHorzBU}(\distHorzBU)\dd\dR\dd\dL\dd \distHorzBU\\
&+\int_{r_\ris}^\infty{\int_{0}^\infty\int_{(\distHorzBU-r_\ris) h_\vehicles/h_\ris}^{\distHorzBU h_\vehicles/h_\bs}} \prob{\SINR>\gamma|\distHorzBU,\dR,\dL}f_{\dR}(\dR)  f_{\dL}(\dL)f_{\distHorzBU}(\distHorzBU) \dd\dR\dd\dL\dd \distHorzBU
\end{align*}
where $(a)$ is due to fact that $\A_\bs$, $\A_\ris$ and $\overline{\A_\bs\cap\A_\ris}$ creates the partition of the sample space, $(b)$ is due to \eqref{eq:Ab}, \eqref{eq:As} and  the fact that $\SINR$ is 0 when $\overline{\A_\bs\cap\A_\ris}$. We now solve the three terms separately.

The first term denotes the coverage under $\A_\bs$ and  $\SINR=\sinrD$. Hence, from \eqref{eq:SI},
\begin{align*}
&\prob{\sinrD>\gamma|\distHorzBU,\dR,\dL}= 
\expects{I}{ 
		\prob{ 
				\frac{\ellBU(\distHorzBU) \gBU u_m v_m}{\noisen  + I} > \gamma |\distHorzBU, I, \dR, \dL
		 }} \\
&= \expects{I|\distHorzBU,\dR,\dL}{ \prob{\gBU  > \frac{\gamma \cdot (\noisen  + I) }{\ellBU(\distHorzBU)u_m v_m}|\distHorzBU, I, \dR, \dL} }\stackrel{(a)}= \expects{I|\distHorzBU,\dR,\dL}{ F_{n_0}{\left(\frac{\gamma \cdot (\noisen  + I) }{\ellBU(\distHorzBU)u_m v_m}\right)} }\\
&\stackrel{(b)}\le \expects{I|\distHorzBU,\dR,\dL}{ 
		\sum_{n=1}^{n_0}
			(-1)^{n+1} \binom{n_0}{n} 
			e^{-n \frac{\gamma \cdot (\noisen  + I) }{\ellBU(\distHorzBU)u_m v_m} \eta_m }
		 }\\
		 &=\sum_{n=1}^{n_0}
			(-1)^{n+1} \binom{n_0}{n} 
			e^{- \frac{ n \gamma \noisen \eta_m}{\ellBU(\distHorzBU)u_m v_m}  }
		\laplaces{I|\distHorzBU,\dR,\dL}{\textstyle \frac{n\gamma\eta_m  }{\ellBU(\distHorzBU)u_m v_m} }.
\end{align*}
Here $(a)$ is due to distribution of $\gBU$ and $(b)$ is due to Alzer's lemma. 
%
%

The second and third terms denote the coverage under $\A_\ris$ with $\SINR=\sinrV$ defined in  \eqref{eq:SIp}. Similar to the first term, it is given as 
\begin{align*}
&\prob{\sinrV>\gamma|\distHorzBU,\dR,\dL}= 
\expects{I}{ 
		\prob{ 
				   \frac{ \ell_{\mathrm{r}}(\distHorzBU, r_\ris) \gUR u_m v_m}{\noisen + I} > \gamma |\distHorzBU, I, \dR, \dL
		 }} \\
		 		 &=\sum_{n=1}^{n_0}
			(-1)^{n+1} \binom{n_0}{n} 
			e^{- \frac{ n \gamma \noisen \eta_m}{\ell_\mathrm{r}(\distHorzBU,r_\ris)u_m v_m}  }
		\laplaces{\intfV|\distHorzBU,\dR,\dL}{\textstyle \frac{n\gamma\eta_m  }{\ell_\mathrm{r}(\distHorzBU,r_\ris)u_m v_m} }.
\end{align*}
Substituting these back in the coverage probability expression, we get the desired result.

%
%

\section{Proof of Lemma \ref{lemma:condistr}}
\label{app:lemma_condistr}

\begin{figure}[ht!]
    \centering
    \includegraphics[width = 0.8\textwidth]{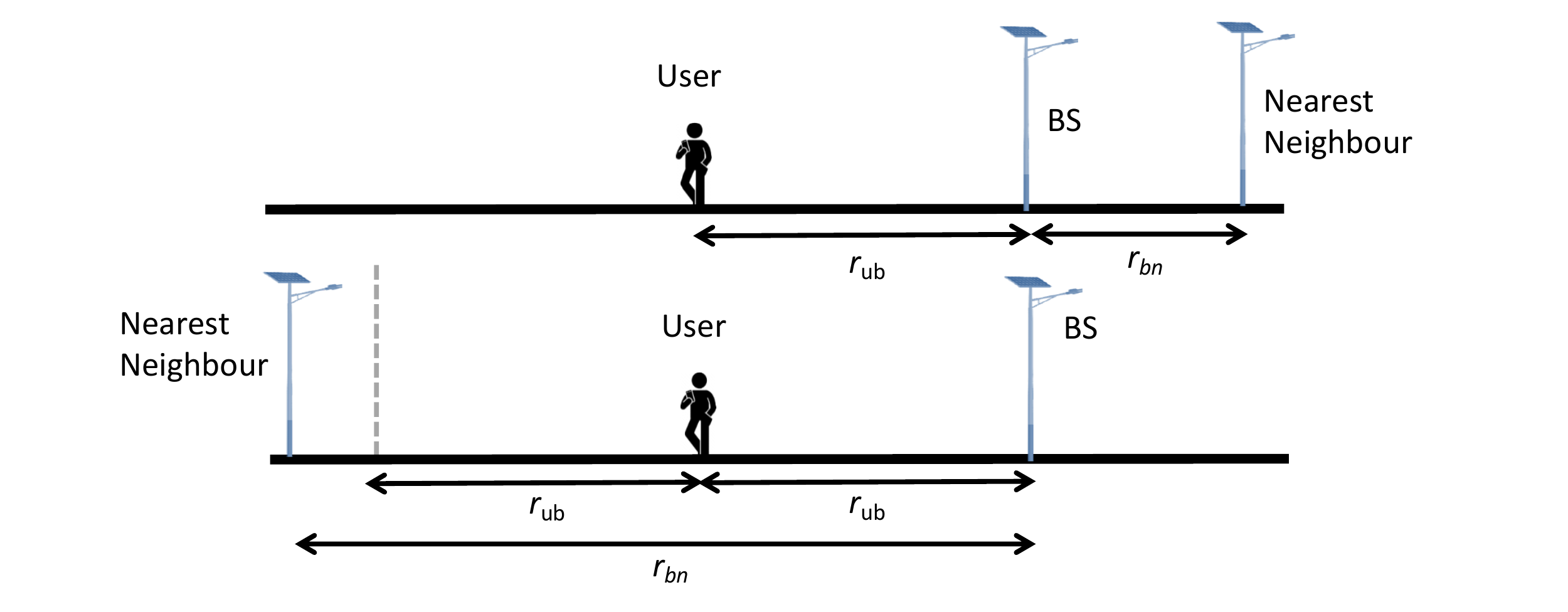}
    \caption{An illustration showing two possible positions of the nearest neighbor of the serving BS.}
    \label{fig:dnn}\figspaceadjust
\end{figure}

We first note that conditioned on the distance $\distHorzBU$ of the closest BS from the user,  no  other BS (including the  BS $X_1$ nearest to the serving BS $X_0$) can be present within the region $[-\distHorzBU,\distHorzBU]$ around the user. Hence, conditioned on $\distHorzBU$, rest of BSs form a PPP with density $\lambda(z)=\lambda_\bs \indside{z \notin [-\distHorzBU,\distHorzBU]  }$. 
Hence,
 \begin{align}
 &\mathbb{P}(\distHorzNN > y | \distHorzBU = x) \nonumber\\
 &= \prob{\text{There is no BS in }[\distHorzBU-y,\distHorzBU+y| \text{the closest BS } X_0 \text{ is at }\distHorzBU}\nonumber\\
 &\stackrel{(a)}=\expS{-\int \lambda(z) \indside{z\in \distHorzBU-y,\distHorzBU+y}  \dd z}\nonumber\\
 &=\expS{-\int \lambda_\bs \indside{z \notin [-\distHorzBU,\distHorzBU]  }\indside{z\in [\distHorzBU-y,\distHorzBU+y]}  \dd z}\nonumber\\
 &=\expS{-\lambda_\bs |[\distHorzBU-y,\distHorzBU+y]\setminus  [-\distHorzBU,\distHorzBU] |}\label{eq:app:Frbn}
\end{align}
where the $(a)$ step is due to void probability of PPP. 

Now if  $0\leq y \le 2\distHorzBU$, then $[\distHorzBU-y,\distHorzBU+y]\setminus  [-\distHorzBU,\distHorzBU]=[\distHorzBU,\distHorzBU+y]$.

If  $\leq y > 2\distHorzBU$, then $[\distHorzBU-y,\distHorzBU+y]\setminus  [-\distHorzBU,\distHorzBU]=[\distHorzBU-y,-\distHorzBU]\cup[\distHorzBU,\distHorzBU+y]$

Substituting these in \eqref{eq:app:Frbn}, we get the desired result.

\begin{acronym}
    \acro{4G}{fourth generation}
    \acro{5G}{fifth generation}
    \acro{AoA}{angle of arrival}
   	\acro{AoD}{angle of departure}
    \acro{BCRLB}{Bayesian CRLB}
    \acro{BS}{base station}
	\acro{CDF}{cumulative density function}
    \acro{CF}{closed-form}
    \acro{CoMP}{coordinated multipoint}
    \acro{CRLB}{Cramer-Rao lower bound}
    \acro{EI}{Exponential Integral}
    \acro{eMBB}{enhanced mobile bstreetband}
    \acro{FIM}{Fisher Information Matrix}
    \acro{GPS}{global positioning system}
    \acro{GNSS}{global navigation satellite system}
    \acro{HetNet}{heterogeneous network}
    \acro{IoT}{internet-of-things}
    \acro{LOS}{line-of-sight}
    \acro{LTE}{long term evolution}
    \acro{MAB}{multi-armed bandits}
    \acro{MBS}{macro base station}
    \acro{MIMO}{Multiple Inputs Multiple Outputs}
    \acro{mm-wave}{millimeter wave}
    \acro{mMTC}{massive machine-type communications}
    \acro{MS}{mobile station}
    \acro{MVUE}{minimum-variance unbiased estimator}
    \acro{NLOS}{non line-of-sight}
    \acro{OFDM}{orthogonal frequency division multiplexing}
    \acro{PDF}{probability density function}
    \acro{PGF}{probability generating functional}
    \acro{PLCP}{Poisson line Cox process}
    \acro{PLT}{Poisson line tessellation}
    \acro{PLP}{Poisson line process}
    \acro{PPP}{Poisson point process}
    \acro{PV}{Poisson-Voronoi}
    \acro{QoS}{Quality of service}
    \acro{RAT}{radio access technique}
    \acro{RIS}{reconfigurable intelligent surface}
    \acro{RSSI}{received signal-strength indicator}
    \acro{RV}{random variable}
    \acro{SBS}{small cell base station}
   	\acro{SINR}{signal-to-interference-plus-noise ratio}
    \acro{SNR}{signal-to-noise ratio}
    \acro{TN}{terminal node}
    \acro{TS}{Thompson sampling}
	\acro{ULA}{uniform linear array}
	\acro{UE}{user equipment}
 	\acro{URLLC}{ultra-reliable low-latency communications}
    \acro{V2V}{vehicle-to-vehicle}  
\end{acronym}


\bibliography{references.bib}
\bibliographystyle{IEEEtran}

\end{document}